\documentclass[a4paper,onecolumn,11pt,accepted=2024-07-03]{quantumarticle}
\pdfoutput=1
\usepackage[utf8]{inputenc}
\usepackage[english]{babel}
\usepackage[T1]{fontenc}
\usepackage{tikz}
\usepackage{lipsum}

\usepackage{amssymb,physics}
\usepackage{hyperref}
\usepackage[all]{hypcap}

\usepackage{amsmath}
\usepackage{amsthm,bm}
\usepackage{xcolor}
\usepackage{appendix}
\usepackage{graphicx}
\usepackage{subfig}

\newcommand{\bl}[1]{ {\color{black} #1} }
\newcommand{\red}[1]{ {\color{black} #1} }

\usepackage[capitalize]{cleveref}

\usepackage{geometry}

\newcommand*{\rom}[1]{\expandafter\@slowromancap\romannumeral #1@}

\theoremstyle{plain}
\newtheorem{theorem}{Theorem}
\newtheorem{lemma}{Lemma}
\newtheorem{remark}{Remark}
\newtheorem{prop}{Proposition}
\newtheorem{definition}{Definition}

\begin{document}

\title{Simulation-assisted Learning of Open Quantum Systems}

\author{Ke Wang}
\email{wangke.math@psu.edu}
\affiliation{Department of Mathematics, Pennsylvania State University}

\author{Xiantao Li}
\email{Xiantao.Li@psu.edu}
\affiliation{Department of Mathematics, Pennsylvania State University}

\maketitle

\begin{abstract}
  Models for open quantum systems, which play important roles in  electron transport problems and quantum computing, must take into account the interaction of the quantum system with the surrounding environment. Although such models can be derived in some special cases, in most practical situations, the exact models are unknown and have to be calibrated. This paper presents a learning method to infer parameters in Markovian open quantum systems from measurement data. One important ingredient in the method is a direct simulation technique of the quantum master equation, which is designed to preserve the completely-positive property with guaranteed accuracy. The method is particularly  helpful in the situation where the time intervals between measurements are large. The approach is validated with error estimates and numerical experiments.  
\end{abstract}

\section{Introduction}
Quantum learning has recently emerged as a field at the intersection of quantum physics and computer science \cite{da2011practical,wang2017experimental,huang2022foundations}.  Its primary focus is on the estimation and characterization of energy levels and interaction coefficients within a quantum system. 
Such effort has become increasingly crucial for characterizing, optimizing, and controlling quantum systems. One particularly explored area is Hamiltonian learning, where the emphasis is to infer the Hamiltonians \red{\cite{granade2012robust,rudinger2015compressed,wiebe2015quantum,burgarth2017evolution,sone2017hamiltonian,wang2017quantum,zubida2022optimal,huang2023learning,haah2022optimal,haah2016sample,bakshi2023learning, ni2024quantum, li2023heisenberg, holzapfel2015scalable}},  which ultimately determine the system's energy spectrum and dynamic behavior. Compared to direct simulation approaches, the learning process attempts to map observation data back to the model parameters. 
An important extension of Hamiltonian learning is the parameter identification of open quantum systems \red{\cite{bairey2020learning,franca2024efficient,samach2022lindblad}}, ones that continuously interact with their environment \cite{breuer2002theory}. The ultimate goal is to reconstruct the interactions in the Lindblad-Gorini-Kossakowski-Sudarshan quantum master equation (QME) \cite{lindblad1976generators,gorini1976completely}:
\begin{equation}\label{eq: lindblad}
    \frac{d}{dt} \rho 
    = -i[H,\rho]+\sum_{j=1}^{N_V}(V_j\rho V^{\dag}_j-\frac{1}{2}V_j^{\dag}V_j\rho-\frac{1}{2}\rho V_j^{\dag}V_j),\\
\end{equation}
where $H$ is the system Hamiltonian and $V_j$'s are jump operators that arise from the interactions with the environment 
\cite{breuer2002theory}. 

In practice, the learning tasks are usually formulated as an optimization problem, which is then solved by an iterative method, and as such, one needs to frequently access the objective function, e.g., through the expectations of some observables.  The ability to acquire such quantities might be limited by the measurement capabilities.   
Furthermore, the availability of the gradient of the objective function, a necessary ingredient for a fast optimization algorithm, might not be attainable through experiments either.
In this paper, we propose a simulation-assisted method to address these issues. In contrast to the equation-based methods \cite{bairey2020learning,franca2024efficient} to identify Lindbladians in \eqref{eq: lindblad}, we formulate a \emph{trajectory-based} framework, where the observations $A(t)=\tr(A e^{t\mathcal{L}} \rho(0) )$ are fit by simulating the QME \eqref{eq: lindblad}. {\color{black} The main departure from equation-based methods \cite{franca2024efficient,rouze2021learning,samach2022lindblad} is that we do not assume access to the quantum system.  We designed a classical learning algorithm with a given time series measurements as input. Namely, we do not assume the flexibility of choosing the time steps or the utilization of classical shadow tomography to make further measurements. For example, the efficiency in the sample complexity  in the approach \cite{franca2024efficient} requires choosing the initial and end time carefully. Part of the challenges comes from the presence of noise in the data, which complicates the numerical differentiation procedure. In our approach, the accuracy is controlled by chopping measurement intervals $\Delta t$ into smaller steps $\delta t$, and in between we simulate the Lindblad dynamics with the proposed parameters. 

 Another notable difference from the equation-based approaches \cite{bairey2020learning,franca2024efficient} is that  our optimization problem does not require the expectations associated with the right-hand side of \cref{eq: lindblad}, i.e., \( \tr(OV_j\rho V_j^\dag), \tr(OV_j^\dag  V_j\rho),  \tr(O\rho V_j^\dag  V_j), \) and thus fewer features are needed from the data set. }

To enable such a simulation-assisted approach, we propose a simulation algorithm, denoted here by $\mathcal{M}_{\delta t}(t)$, that has global error $\delta t^2$, in that $e^{\mathcal{L}  t} - \mathcal{M}_{\delta t}(t) = \mathcal{O}\left(\delta t^2\right)$, we choose a semi-implicit algorithm, which is  stable even when the coherent term in the QME \eqref{eq: lindblad} has large coefficients. This makes it more robust in practice than the second-order method in Breuer and Petruccione \cite{breuer2002theory}. Furthermore, we also show that the method has a completely positive property and it can be easily written in a Kraus form. This makes it possible to implement such a simulation method on a quantum computer as well. This is done by unraveling the Lindblad dynamics to a stochastic Schr\"odinger equation \cite{breuer2002theory}, followed by a stochastic expansion \cite{kloeden2011numerical}. 

Another practical aspect of our approach is efficient optimization. With the Kraus representation of our numerical approximation, the calculation of the gradient is streamlined. This allows us to apply the Levenberg-
Marquardt algorithm, which has very rapid convergence \cite{yamashita2001rate,fan2005quadratic}. With mild assumptions on the fitting error, we will show how the approximation error from the numerical simulation and the statistical error from the measurements affect the performance of the parameter identification. 

The rest of the paper is organized as follows: We present a general learning framework 
in the form of a nonlinear least squares  problem in \cref{sec:nlsq}, and explain the difference and connections to existing works in \cref{sec: related}. To emphasize the algorithmic details, we present the methods without reference to specific open quantum systems, also with the hope that the methods can be applied to a broader class of problems. Then in \cref{sec: simulation} and \cref{sec: optm}, we present the specific simulation method and how it is integrated with the optimization procedure. In \cref{sec: error} and \cref{sec: num}, we present some error analysis and results from some numerical experiments, \red{where we detail the specific applications of open quantum systems.}

{\color{black}
\section{Notations and Terminologies}
Throughout this paper, the Euclidean norm is denoted by $\|\cdot \|$, i.e. for $\bm v \in \mathbb{C}^d,$ $\|\bm v\| = (\sum_{i=1}^d v_i^2)^{1/2}$. 
The density operator $\rho \in \mathbb{C}^{d\times d}$  is represented as a semi-positive definite matrix   with trace $1$, properties that will be expressed as $\rho \succeq 0$ and $\tr(\rho)=1.$  An emphasis will be placed on quantum evolutions that are trace-preserving and completely positive. These properties are formalized in terms of dynamic maps $\mathbb{C}^{d\times d} \to \mathbb{C}^{d\times d}$, e.g. see \cite{watrous2018theory}. In particular, if  $\tr(\mathcal{A}(\rho)) = \tr(\rho)$ for all $\rho$, then  $\mathcal{A}$ is said to be trace-preserving. Similarly, $\mathcal{A}$ is said to be positive if $\rho \succeq 0 \Rightarrow \mathcal{A} \rho \succeq 0$. Further, $\mathcal{A}$ is said to be completely positive if $\mathcal{A}\otimes I_{m\times m} $ is a positive map for every $m \ge 1$.
A useful representation of trace preserving and completely positive  (CPTP) maps is the Kraus form, which expresses $\mathcal{A}$  as follows, 
\begin{equation}
    \mathcal{A}\rho = \sum_j V_j\rho V_j^{\dag}, \quad \mathrm{with}\;  \sum_j V_j^{\dag} V_j = I
\end{equation}

In this paper, $\tau_m$ and $t_n$ respectively denote the simulation and measurement times. The corresponding time interval is represented by $\delta t$ and $\Delta t$, respectively, with the ratio denoted by $L = \frac{\Delta t}{\delta t}$, see \cref{fig:NotationsForTime}. In particular, the simulation times are designed to be shorter or equal to the measurement times. The ratio $L$ can be controlled to obtain the desired accuracy.}
\begin{figure}[h]
    \centering
    \includegraphics[scale = 0.4]{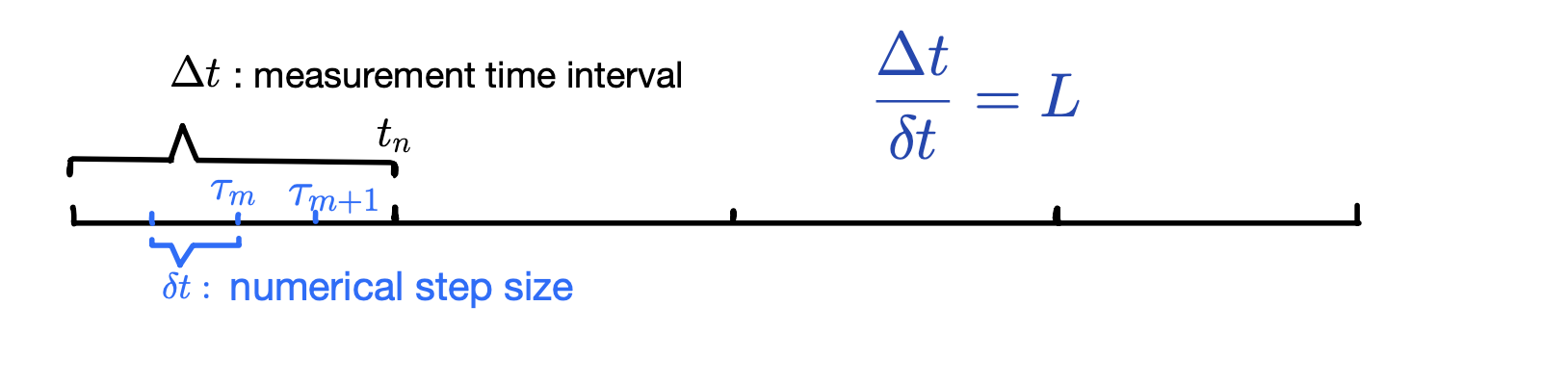}
    \caption{Time steps associated with the simulations and measurements. The measurement time interval $\Delta t$ is chopped into smaller intervals with  $\delta t$ representing numerical step size. }
    \label{fig:NotationsForTime}
\end{figure}

\section{The Learning Framework for Parameter Identification of Lindbladians}

\subsection{A Least Squares Formulation of the Learning Problem}\label{sec:nlsq}
In practice, the unknown parameters in an open quantum system can appear in both the Hamiltonian and the dissipative terms in  the QME \eqref{eq: lindblad}. To indicate the role of the parameters, we first rewrite \cref{eq: lindblad} as follows,
\begin{equation}\label{eq: Lindblad1}
\begin{aligned}
    \frac{d}{dt} \rho  &= \mathcal{L}_H(\bm \theta_H) \rho + \mathcal{L}_D(\bm \theta_D) \rho, \\
    \mathcal{L}_H(\bm \theta_H)\rho = -i[H,\rho], &
   \quad \mathcal{L}_D(\bm \theta_D)\rho = \sum_{j=1}^{N_V}\left(V_j\rho V_j^{\dag}-\frac{1}{2}V_j^{\dag}V_j\rho-\frac{1}{2}\rho V_j^{\dag}V_j\right).
    \\
\end{aligned} 
\end{equation}
We can combine the parameters by writing $\bm \theta:= (\bm \theta_H, \bm \theta_D)$; $\mathcal{L}(\bm \theta):= \mathcal{L}_H(\bm \theta_H) + \mathcal{L}_D(\bm \theta_D)$ will be called the  \emph{Lindbladian}.  We refer to the two sets of parameters as Hamiltonian and dissipative parameters, respectively. When $\mathcal{L}_D=0$, the problem is reduced to a Hamiltonian learning problem.  

We assume that the measurement data is a time series $y_{k,n}, n= 1,2,\cdots, N_T,$ which correspond to observables $A^{(k)}, k =1,\cdots,N_O$, at time instances $t_n=n\Delta t, \quad 0\leq n \leq N_T,$
\begin{equation}\label{ykn}
   y_{k,n} := \tr \big( A^{(k)}  \rho(t_n ;\bm \theta) \big).
\end{equation}
We have set $t_0=0,$ and for simplicity we assume a uniform time-lapse $\Delta t$ between experiments, although the extension to general distributions of measurement times is straightforward. Also indicated in \cref{ykn} is the dependence of the solution of \eqref{eq: Lindblad1} on the parameters,  i.e., $\rho(t_n ;\bm \theta) = e^{t_n \mathcal{L}(\bm \theta) } \rho(0).$

We now formulate the learning problem based on the consistency in \cref{ykn}, as a least squares problem,
\begin{equation}\label{eq: minphi}
    \min_{\bm \theta} \phi(\bm \theta),
\end{equation}
where the objective function is given by,
\begin{equation}\label{eq: phi}
    \phi(\bm \theta) = \frac1{2N_O N_T} \sum_{\overset{k=1,\cdots,N_O}{n = 1,\cdots,N_T}} \abs{  y_{k,n}-y_{k,n}^* }^2, \quad y_{k,n}^* := \tr \big( A^{(k)}  \rho(t_n ;\bm \theta^*) \big),
\end{equation}
Here, $y_{k,n}^* $ is interpreted as the measurement data with respect to the true, but unknown parameter $\bm \theta^*$; $\bm\theta^*$ is the solution of the least squares problem, i.e., $\bm \theta^* = \mathrm{argmin} \phi(\bm\theta)$.
Due to the typical nonlinear dependence of $ \rho(t_n ;\bm \theta)$ on $\bm \theta$, the optimization problem is a nonlinear least squares problem \cite{kelley1999iterative}.

An alternative viewpoint is a maximum likelihood estimate (MLE)  from parameter estimation methods for dynamical systems \cite{brunel2008parameter}. Namely, the measurement outcome $y_{k,n}$ comes  with a Gaussian noise $\epsilon_{k,n}$,
\[ 
 y_{k,n} = \text{tr} \big( A^{(k)}  \rho(t_n ;\bm \theta) \big) + \epsilon_{k,n}.
\]
Then \cref{eq: phi} is the log-likelihood function.

\subsection{Related Works}\label{sec: related}
The QME \eqref{eq: lindblad} is an important alternative to study electron transport properties, which are important in material science \cite{biele2012stochastic,di2007stochastic}. \red{Attempts have been made to integrate the model with existing parallel code \cite{appel2009stochastic} to simulate systems with many electrons. These simulation algorithms can be regarded as a bottom-up approach, where the bath operators are assumed in advance. }
 Another increasingly important application is quantum computing, where the quantum circuits are subject to environmental noise, and many quantum error mitigation schemes require the knowledge of an error model, i.e., the channel  interactions and strength \cite{temme2017error,endo2018practical,cai2023quantum}.

  Bairey et al.  \cite{bairey2020learning} proposed to learn the coefficients in the QME \eqref{eq: lindblad} by acting the steady-state equation on a collection of observables. These equations can thus be rewritten in terms of the expectations by applying the adjoint of the Lindbladians to the observables, leading to equations of Ehrenfest form.  For example, by letting $\langle A\rangle(t) := \text{tr}(A\rho(t))$, one can derive from \cref{eq: lindblad}, 
\begin{equation}
\label{eq:measure dynamics}
    \frac{d}{dt}\langle A\rangle (t)= -i\langle[A,H]\rangle(t)+\sum_{j=1}^{N_V}\left(\langle V_j^{\dag}AV_j\rangle(t)-\frac{1}{2}\langle AV_j^{\dag}V_j\rangle(t)-\frac{1}{2}\langle V_j^{\dag}V_jA\rangle(t) \right).
\end{equation}
The operators on the right-hand side can be expressed on a known basis with unknown parameters, which reformulate the equation in terms of measurement data and parameter values. 
At a steady state, the time derivative drops out, and the approach by Bairey et al.  \cite{bairey2020learning} yields a system of equations, usually over-determined, for the unknown parameters.  
Franca et al.  \cite{franca2024efficient} extended this framework to dynamics problems, with the time derivatives estimated from polynomial approximations. As alluded to in the previous section, these two approaches can be summarized as \emph{equation-based:} in that the objective function is set up so that the equation  \eqref{eq:measure dynamics} holds the best it can.  An important practical issue is that the time interval for the measurements, denoted in this article by $\Delta t$,  might be large, in which case,  the accuracy of the inference is limited by the approximation of the time derivatives.  \red{  If there is no means to decrease the sample time interval, we face inherent inaccuracies in the learning method.  More importantly, the measurement data in practical applications always contain noise, leading to a subtle challenge: For the numerical differentiation procedure to be accurate, $\Delta t$ must be sufficiently small. However, a small $\Delta t$ can amplify the effect of measurement noise. This difficulty has been identified and analyzed in other applications \cite{jauberteau2009numerical, ahnert2007numerical}.  The approach in \cite{franca2024efficient} mitigated this issue by choosing the measurement times based on the Chebyshev measure, which on the other hand, requires an upper bound on the time interval $[0, t_{\max}].$ Another challenge is that the right-hand side of Eq. \eqref{eq:measure dynamics} may involve many expectations that have to be obtained from measurements.} It is also worthwhile to mention that 
 Boulant et al. \cite{boulant2003robust} also proposed numerical algorithms for reconstructing Lindblad operators and highlighted the importance of the CP property. But their method to ensure this property is quite involved.

The more recent work \cite{av2020direct} aims at reconstructing Lindbladians based on measurements at multiple times with an optimization problem.   Their method is based on experimental measurements and the least squares problem is formulated in terms of the discrete probability induced by the observables. Numerical simulations are only used to determine a stopping criterion. Compared to this experiment-based approach, the simulation-assisted approach only works with the original set of measurement data. More importantly, the numerical simulation provides the gradient of the objective function, which can substantially speed up the optimization process. 
\red{Another related approach to Lindbladian learning is the learning algorithms from Gibbs states~\cite{haah2022optimal,bakshi2023learning}, since the Gibbs state could be reached by certain Lindblad dynamics.
 }

While the equation-based methods \cite{bairey2020learning,franca2024efficient}  fundamentally differ from the method presented in this paper, the two approaches can be combined to accelerate the parameter estimation procedure. Indeed, the current problem can be viewed, in the broader context, as parameter estimation of  ODEs \cite{brunel2008parameter}, where the polynomial approximation of derivatives is known as polynomial regression. It has been used as a preparation of a two-step estimation procedure \cite{strebel2013preprocessing,bhaumik2015bayesian}.  The parameters obtained from the linear regression, e.g., in the approach of Franca et al. \cite{franca2024efficient} can be used as an initial guess for a trajectory-based approach.

\subsection{Computing the Objective Function using Direct Simulations of the QME}\label{sec: simulation}

Solving the optimization problem in \eqref{eq: minphi} requires access to the expectation of $A^{(k)}$ for an approximate parameter $\bm \theta$, and such data are unavailable from experiments.  This is treated by an efficient numerical algorithm for solving the QME  \eqref{eq: Lindblad1}. Here we outline a simple derivation of numerical methods, which can be implemented on either classical or quantum computers.

In a nutshell, to approximate the solution operator $e^{\mathcal{L}t}$ of the QME \eqref{eq: lindblad},  a simulation-assisted algorithm uses an approximation  $\mathcal{M}_{\delta t}(t)$  and minimizes the difference between $e^{\mathcal{L}^\dagger (t_n)} A$ and $\mathcal{M}_{\delta t}^\dagger(t_n) A$, by taking the trace with $\rho(0).$ Another advantage of this approach is that we no longer have to measure the terms on the right-hand side, as was done in \cite{bairey2020learning,franca2024efficient}. Thus the number of measurements can be significantly reduced. 

The QME encodes a trace-preserving and completely positive (TPCP) map \cite{lindblad1976generators,gorini1976completely}, which in the Kraus form, can be generally written as \cite{watrous2018theory} (Choi-Kraus’theorem),
\begin{equation}\label{tpcp}
    \rho_{n+1} = \mathcal{K}[\rho_n]:= \sum_{j} F_j \rho_n F_j^{\dagger},
\end{equation}
where $\sum_j F_j^{\dag}F_j = I$\cite{breuer2002theory}. Therefore, if such a Kraus form associated with the solution $e^{t_n \mathcal{L}(\bm \theta) }$ of the QME \eqref{eq: Lindblad1} can be found, then the objective function in \cref{eq: phi} and its gradient can be directly evaluated to carry out the optimization task. Another interesting aspect is that this procedure also places the problem in the general framework of learning a quantum system \cite{huang2022foundations}. Boulant et al. \cite{boulant2003robust} have demonstrated that the CP property increases the overall robustness of the parameter estimation procedure.

Our approach starts by first unraveling the Lindblad equation \eqref{eq: Lindblad1} into the stochastic Schr\"odinger equation \cite{breuer2002theory},
\begin{equation}
\label{eq:SSE}
    \mathrm{d}\ket{\psi(t)} = \left(-iH \ket{\psi(t)}-\frac{1}{2}\sum_{j=1}^{N_V} V_j^{\dag}V_j\ket{\psi(t)} \right)\mathrm{d}t +\sum_{j=1}^{N_V} V_j\ket{\psi(t)} \mathrm{d}W_{j;t}
\end{equation}
Here $\ket{\psi(t)}$ represents a stochastic realization of a quantum state $\rho(t)$, in the sense that,
\begin{equation}\label{eq: psi2rho}
    \rho(t) = \mathbb{E}\big[\ketbra{\psi(t)}\big].
\end{equation}
In the stochastic equation, $W_{j;t}$'s are independent Brownian motions that incorporate the noise from the bath.

Although in practice the trace-preserving property can be ensured by simple scaling, the completely positive property is often destroyed by classical ODE methods \cite{cao2021structure}. \bl{A simple example is the Euler's method, e.g., applied to a simple Lindblad equation: $\frac{d}{dt}\rho = - \frac{1}{2} \{V^\dag V, \rho\} + V \rho V^\dag, $
\[
\rho_{n+1} = \rho_n - \frac12 \delta t  V V^\dag \rho_n - \frac12  \delta t\rho_n V^\dag V +  \delta t V \rho_n V^\dag. 
\]
The right hand side is a Kraus form, but not in a diagonal form: It can be 
written as, $\rho_{n+1} = \sum_{i,j=1}^3 a_{i,j} F_i \rho_n F_j^\dag$, with
$F_1=I, F_2= V, F_3=V V^\dag $. But the corresponding matrix $(a_{i,j})$ is clearly  not positive definite. Consequently, standard ODE methods may not induce a CP map. }

The equivalence between the QME \eqref{eq: lindblad} and \cref{eq: psi2rho} is the key ingredient for constructing an approximation of the QME that preserves the CP property.   As a quick demonstration, we first consider a time discretization of the stochastic Schr\"odinger equation \eqref{eq:SSE}  by the semi-implicit Euler method \cite{kloeden2011numerical}. Toward this end, we define numerical time steps, $\tau_m=m\delta t, m= 0, 1, \cdots, M$, and an approximate wave function $\ket{\psi_m}$,
\begin{equation}
    \ket{\psi(\tau_m)} \approx \ket{\psi_m}, \; m\geq 0. 
\end{equation}
It is important to point out that the step size $\delta t$ is a numerical parameter, and it can be much smaller than the measurement time $\Delta t.$ In order for the simulation to produce expectations at the same time steps as the experiments, choose $\delta t$ such that,
\begin{equation}\label{lmn}
    L= \frac{\Delta t}{ \delta t} \in \mathbb{N}, \quad M= L N_T.
\end{equation}

The semi-implicit Euler method follows a time-marching step and implements an iteration formula for $\ket{\psi_m}$ as follows,
\begin{equation}\label{eq: semi-euler}
\ket{\psi_{m+1}} = \ket{\psi_m}+G\ket{\psi_{m+1}}\delta t+\sum_{j=1}^{N_V} V_j\ket{\psi_m}\delta W_{j,m}. 
\end{equation}
The terminology "semi-implicit" comes from the treatment of the noise: It is only sampled from the current time step and if $V_j=0$, this method is the standard implicit Euler method for an ordinary differential equation.
In \cref{eq: semi-euler} $\delta W_{j,m}$'s are independent Gaussian random variables with mean zero and variance $\delta t.$ In addition, 
 the matrix $G$ is given by,
\begin{equation}\label{matG}
    G= -iH-\frac{1}{2}\sum_{j=1}^{N_V} V_j^{\dag}V_j.
\end{equation} 
It also appears in the structure-preserving scheme \cite{cao2021structure}.

We can express the method in a compact form
\begin{equation}\label{eq: semi-euler1}
    \ket{\psi_{m+1}} = (I-G\delta t)^{-1}\left[ I +\sum_{j=1}^{N_V} V_j \delta W_{j,m} \right] \ket{\psi_m}.
\end{equation}
At this point, an approximation method for the density operator $\rho$ can easily be obtained.  Specifically, let $\rho_m = \mathbb{E}\big[\ket{\psi_{m}}\bra{\psi_{m}}\big], m\geq 0$; $\rho(\tau_m) \approx \rho_m.$  By taking expectations of \eqref{eq: semi-euler1}, we arrive at, 
\begin{equation}\label{semi-Euler}
\rho_{m+1} = (I-G\delta t)^{-1}\rho_m(I-G\delta t)^{-\dag}+(I-G\delta t)^{-1}\sum_{j=1}^{N_V} V_j\rho_m V_j^{\dag}(I-G\delta t)^{-\dag}\delta t.
\end{equation}
Here we have used $\mathbb{E}\big[ \delta W_{j,m}]=0$ and $\mathbb{E}\big[ \delta W_{j,m} \delta W_{k,m} \big]= \delta t \delta_{j,k}.$

Clearly, this can be written in the Kraus form  \eqref{tpcp}. The corresponding Kraus operators are given by,
\begin{align}\label{F0Fj}
 F_0 =& (I-G\delta t)^{-1},  \\ F_j =& (I-G\delta t)^{-1} V_j\sqrt{\delta t}, j= 1,\cdots,N_V.
\end{align}

 This Kraus form from the semi-implicit Euler method is summarized in \cref{lemma:1.0}.
\begin{lemma}
\label{lemma:1.0}
    The  semi-implicit Euler method induces a Kraus form, i.e.,
    $\rho_{m+1} = \mathcal{K}[\rho_m] = \sum_j F_j\rho_m F_j^{\dag}$ where $\sum_j F_j^{\dag}F_j = I+\mathcal{O}(\delta t^2).$ In addition,
    \begin{equation}
        e^{\delta t \mathcal{L}(\bm \theta) }\rho  - \mathcal{K}\rho = \mathcal{O}(\delta t^2). 
    \end{equation}
\end{lemma}
The approximation property can be verified by direct expansions of the Kraus operators in \cref{F0Fj}. 

In practice, such first-order methods have limited accuracy. To ensure better performance, we consider the second-order implicit approximation method for stochastic differential equations \cite[Chapter 15]{kloeden2011numerical}. When applied to \cref{eq:SSE}, the method can be written as,
\begin{equation}\label{eq: 2nd-order}
\begin{aligned}
    \ket{\psi_{m+1}} &= (I-\frac{1}{2}G\delta t)^{-1}(I+\frac{1}{2}G\delta t)\ket{\psi_m}+\sum_{j=1}^{N_V}(I-\frac{1}{2}G\delta t)^{-1}(V_j+\frac{1}{2}V_jG\delta t)\ket{\psi_m}\delta\hat{W}_j\\
    &+\frac{1}{2}(I-\frac{1}{2}G\delta t)^{-1}\sum_{j_1,j_2 = 1}^{N_V} V_{j_2}V_{j_1}\ket{\psi_m}\left(\delta \hat{W}_{j_1,m}\delta\hat{W}_{j_2,m}+U_{j_1,j_2,m}\right).
\end{aligned}
\end{equation}

In this expression,  $\delta\hat{W}$ is an approximation of an increment of the Brownian motion. 
In addition, $U_{j_1,j_2}$ are independent two-point distribution, defined as,
\begin{equation}
\begin{aligned}
    &P(U_{j_1,j_2} = \pm\delta t) = \frac{1}{2},\forall j_2=1,\cdots,j_1-1,\\
&U_{j_1,j_1} = -\delta t,\\
&U_{j_1,j_2} = -U_{j_2,j_1}, \text{ for } j_2 = j_1+1,\cdots,N_V. j_1=1,\cdots,N_V,\\
\end{aligned}
\end{equation}

We now state the approximation property of this method. 

\begin{lemma}
\label{lemma: 2.0}
    The implicit second-order approximation induces a Kraus form, i.e. $\rho_{m+1} = \sum_j F_j\rho_m F_j^{\dag}$, where $\sum_j F_jF_j^{\dag} = I+\mathcal{O}(\delta t^3)$.  In addition, the one-step error is given by, 
    \begin{equation}
        e^{\delta t \mathcal{L}} \rho = \sum_j F_j\rho F_j^{\dag} + \mathcal{O}(\delta t^3).
    \end{equation}

\end{lemma}
\begin{proof}
    Similar to Lemma 1, we define $\rho_m = \mathbb{E}[\ket{\psi_m}\bra{\psi_m}]$. 
Following the second-order implicit scheme \eqref{eq: 2nd-order}, each iteration  of  the density operator is as follows,
\begin{equation}\label{2ndorder}
\begin{aligned}
     \rho_{m+1} &= (I-\frac{1}{2}G\delta t)^{-1}(I+\frac{1}{2}G\delta t)\rho_m(I+\frac{1}{2}G\delta t)(I-\frac{1}{2}G\delta t)^{-1}   \\
     & +(I-\frac{1}{2}G\delta t)^{-1}\sum_{j=1}^{N_V}V_j(I+\frac{1}{2}G\delta t)\rho_m(I+\frac{1}{2}G\delta t)V_j^{\dag}(I-\frac{1}{2}G\delta t)^{-1}\delta t\\
     &+\frac{1}{2}(I-\frac{1}{2}G\delta t)^{-1}\sum_{j_1,j_2=1}^{N_V} V_{j_1}V_{j_2}\rho_m V_{j_2}^{\dag}V_{j_1}^{\dag}(I-\frac{1}{2}G\delta t)^{-1}(\delta t)^2\\
     & = \sum_{j=0}^{N_V^2+N_V}F_j\rho_m F_j^{\dag}=: \mathcal{K}(\rho_m)
\end{aligned}
\end{equation}
where the Kraus operators are given by, 
\begin{equation}
    \begin{aligned}
         F_0 &= (I-\frac{1}{2}G\delta t)^{-1}(I+\frac{1}{2}G\delta t),\\
        F_j & = (I-\frac{1}{2}G\delta t)^{-1}V_j(I+\frac{1}{2}G\delta t)\sqrt{\delta t}, j = 1,\cdots,N_V,\\
        F_{j_1+N_Vj_2} & = \frac{1}{\sqrt{2}}(I-\frac{1}{2}G\delta t)^{-1} V_{j_1}V_{j_2}\delta t, \quad j_1,j_2 =1,\cdots, N_V.
    \end{aligned}
\end{equation}

Notice that the one-step error is of the third order, so we can simplify the iteration formula without changing the error order by letting
\begin{equation}
\label{eq: simplified2order}
    \begin{aligned}
        F_{j_1+N_Vj_2}=\frac{1}{\sqrt{2}} V_{j_1}V_{j_2}\delta t,\quad j_1,j_2 =1,\cdots, N_V.
    \end{aligned}
\end{equation}
This will simplify the calculation of the gradient. The rest of the proof is given in \cref{appendix: 2.0 scheme}.

\end{proof}

We chose an implicit method due to its numerical stability. For example, one can show that the  spectral radius  of each Kraus operator has an upper bound that is independent of $H.$ This can be seen from the observation that the real part of the matrix $G$ is positive definite.  Therefore, the spectral radius of $F_0$ is less or equal to 1. For $0<j\leq N_V$, $F_j$  has  a spectral radius less or equal to the spectral radius of $V_j \sqrt{\delta t}$. \bl{The matrix inverse of $I-\frac{1}{2}G\delta t$ will certainly complicate the numerical implementation. One robust approach is to solve the associated linear system of equations by bi-conjugate gradient method (Bi-CGSTAB) \cite{sleijpen1993bicgstab}, which involves matrix-vector multiplication and orthogonalization.  On the other hand, if $H$ does not involve large eigenvalues, an explicit method can be used to replace \cref{eq: semi-euler,eq: 2nd-order} to simplify the implement, e.g., using the second-order It\^o-Taylor expansion \cite{kloeden2011numerical}.
 }

\subsection{Gradient Evaluations for Gradient-based Optimization}\label{sec: optm}

The problem of finding the optima of $\phi(\bm \theta)$ can be seen as the nonlinear least squares problem
\begin{equation}\label{sq}
    \phi(\bm \theta) = \frac{1}{2N_O N_T}\sum_{\overset{k=1,\cdots,N_O}{n = 1,\cdots,N_T}}|r_{k,n}(\bm \theta)|^2 = \frac{1}{2N_O N_T} R(\bm \theta)^TR(\bm \theta),
\end{equation}
where we used the textbook notations \cite{kelley1999iterative}:  $r_{k,n} = y_{k,n}-y_{k,n}^*$ is the residual error and   $R(\bm \theta) = (r_{k,n})_{k,n}$ will be referred to as \emph{the residual} vector. We are interested in the scenario when the dimension of the residual is larger than the number of parameters, i.e., the over-determined regime. 

One practical iteration method for solving the least squares  problem is the 
Levenberg-Marquardt (LM) algorithm, which, starting with an initial guess, $\bm \theta^{(0)} $, involves updating the parameters iteratively,
\begin{equation}\label{alg:LM}
   \bm \theta^{(k+1)}-\bm \theta^{(k)} = -\left(\nu_k I+R'(\bm \theta^{(k)})^TR'(\bm \theta^{(k)})\right)^{-1}R'(\bm \theta^{(k)})^TR(\bm \theta^{(k)}).
\end{equation}
It can be seen as a combination of the Gauss-Newton method and the gradient descent algorithm. In \cref{alg:LM}, $R'$ denotes the gradient of the residual with respect to the parameters. The parameter $\nu_k$ serves as a regularization, and in practice, it can be chosen to be proportional to the norm of the residual error.

The convergence of the LM method  has been established under very mild conditions. Here we follow \cite{yamashita2001rate,fan2005quadratic} and pose the following local condition: There exists a constant $C,$ such that,
\begin{equation}\label{ident}
    \norm{R(\bm \theta)} \geq C \norm{\bm \theta - \bm \theta^*}, 
\end{equation}
for all $\bm \theta $ in a neighborhood of $\bm \theta^*$. 
\begin{theorem} \cite[Theorem 2.1]{yamashita2001rate}
Assume that \cref{ident} holds and $R'$ is Lipschitz continuous.  Let the  parameter $\nu_k$ be \[\nu_k = \norm{R (\bm \theta^{(k)})}^2.\]  If the initial guess $\bm \theta^{(0)}$ is sufficiently close to $\bm \theta^*$, then there exists a constant, such that the iterations from \eqref{alg:LM} satisfy quadratic convergence,
\begin{equation}
    \norm{\bm \theta^{(k+1)}-\bm \theta^*} \leq C  \norm{\bm \theta^{(k)}-\bm \theta^*}^2.
\end{equation}
\label{theorem: quadratic convergence}
\end{theorem}
The convergence speed is faster than the super-linear convergence proved in \cite{kelley1999iterative}.  

The quadratic convergence property makes the LM method \eqref{alg:LM} an extremely useful algorithm. Global convergence has also been analyzed in \cite{yamashita2001rate,fan2005quadratic}. This is often accomplished by using a line search algorithm.   
\bl{
\begin{remark}
    The LM algorithm's fast convergence comes with the cost of computing the inverse of a matrix in \eqref{alg:LM}. One alternative, which is particularly useful for large scale problems,  is the gradient descent algorithm, 
    \[
     \bm \theta^{(k+1)}-\bm \theta^{(k)} = -\alpha_k R'(\bm \theta^{(k)})^TR(\bm \theta^{(k)}),
    \]
    where $\alpha_k$ is the learning rate.  Such an algorithm can often achieve linear convergence in a neighborhood of the minimizer \cite{nesterov2013introductory}.
\end{remark}
}

In parameter estimation problems, the condition in \cref{ident}  is viewed as a local identifiability condition \cite{brunel2008parameter}. Meanwhile, the Lipschitz condition can be verified by assuming that
the first order and second derivatives of $\mathcal{L}(\bm \theta)$ are bounded, which holds trivially if the Lindbladian has a linear dependence on the parameters.  
To elaborate on this, a partial derivative of the residual error is given by,
    \begin{equation}
       \frac{\partial}{\partial {\theta_\alpha} }  r_{k,n}(\bm \theta) = \text{tr}(A^{(k)}\partial_{\theta_\alpha}\rho(t_n;\bm \theta)). 
    \end{equation}
This leads us to consider  $\Gamma_\alpha(t,\bm \theta) := \partial_{\theta_\alpha}\rho(t;\bm \theta),$ which from \cref{eq: Lindblad1}, follows the differential equation,
\[
 \frac{d}{dt} \Gamma_\alpha(t,\bm \theta) = \mathcal{L}(\bm\theta) \Gamma_\alpha(t,\bm \theta) +  \frac{\partial}{\partial {\theta_\alpha} } \mathcal{L}(\bm \theta) \rho(t,\bm \theta). 
\]
Using the contraction property of $e^{t\mathcal{L}}$ \cite{ruskai1994beyond}, we find that, 
\[
\norm{\Gamma_\alpha(t,\bm \theta) } \leq t \norm{\partial_{\theta_\alpha}  \mathcal{L}(\bm \theta) }.
\]
To examine the Lipschitz continuity of the partial derivatives, we define,
\[ \chi_{\alpha,\beta} : =  \partial_{\theta_\beta} \Gamma_\alpha(t,\bm \theta), \]
which follows the equation, 
\[
 \frac{d}{dt} \chi_{\alpha,\beta} (t,\bm \theta) = \mathcal{L}(\bm\theta) \chi_{\alpha,\beta}(t,\bm \theta) +  \frac{\partial^2}{\partial \theta_\beta \partial \theta_\alpha  } \mathcal{L}(\bm \theta) \rho(t,\bm \theta) +   \frac{\partial}{\partial {\theta_\alpha} } \mathcal{L}(\bm \theta) \Gamma_\beta(t,\bm \theta) 
 +  \frac{\partial}{\partial {\theta_\beta} } \mathcal{L}(\bm \theta) \Gamma_\alpha(t,\bm \theta),
\]
and a simple bound follows,
\[
\norm{\chi_{\alpha,\beta}(t,\bm \theta)} \leq t \norm{\frac{\partial^2}{\partial \theta_\alpha \partial \theta_\beta  } \mathcal{L}(\bm \theta)} + \frac{t^2}{2} 
\big( 2 \norm{\partial_{\theta_\alpha}  \mathcal{L}(\bm \theta) } + \norm{\partial_{\theta_\beta}  \mathcal{L}(\bm \theta) } \big).
\]
As a result, if the first and second order derivatives of $\mathcal{L}(\bm\theta)$ are  bounded, $R'$ fulfills the Lipschitz condition.

\bigskip 


\medskip 

\bl{
Meanwhile, in our learning task, the residual function is subject to approximation error.  To incorporate these errors, we can follow the proof in \cite{yamashita2001rate}, where the search direction at each step of the LM algorithm involves a linear least-square (LS) problem. Based on the classical sensitivity analysis for LS \cite{van1974stability}, an $\mathcal{O}(\epsilon)$ perturbation of $R(\bm \theta)$ and $R'(\bm \theta)$ will introduce an $\mathcal{O}(\epsilon)$ error in the iteration formula. Therefore we have       }
\begin{prop}
Let $\epsilon >0.$    Suppose that the residual vector $R$ and $R'$ is subject to an $\epsilon$ perturbation, \[ R \to R+\Delta R,\] with $\norm{\Delta R} \leq \epsilon$ and $\norm{\Delta R'} \leq \epsilon$,   
    for all $\bm \theta$ in a neighborhood of $\bm \theta^*$.
    Then the LM iterations, under the same conditions on $R$ as in the previous theorem,  exhibit an approximate quadratic convergence,
    \begin{equation}
    \norm{\bm \theta^{(k+1)}-\bm \theta^*} \leq C  \norm{\bm \theta^{(k)}-\bm \theta^*}^2 + \epsilon.
\end{equation}
\bl{where $C$ is independent of $\epsilon$.}
\end{prop}

\medskip

The LM algorithm in \cref{alg:LM} requires access to the gradient of the residual function. To facilitate the calculation of the gradient, we first take the derivative of the Kraus form,

\begin{equation}
    (\partial_{\theta}\mathcal{K})[\rho]  = \sum_{j}\partial_{\theta}F_j\rho F_j^{\dag}+F_j\rho\partial_{\theta}F_j^{\dag}
\end{equation}
Here $\theta$ refers to one parameter in $\bm \theta$. The entire gradient can be computed by visiting all the components in $\bm \theta$. For clarity, we express the results as an inner product in the space of Hermitian matrices. Namely, for any Hermitian matrices $A$ and $B$, we define,
\begin{equation}
    \big\langle A, B \big\rangle:= \tr(AB).
\end{equation}
Now we show how the gradient of the objective function can be computed from a back propagation procedure. 

\begin{lemma}
Assume that the iteration $\rho_{m} \to \rho_{m+1} $ follows a Kraus form,
\begin{equation}
    \rho_{m+1} = \mathcal{K} \rho_m:= \sum_{j} F_j \rho_m F_j^{\dagger}.
\end{equation}
 Then, for any Hermitian operator $A$, the following identify holds,
\begin{equation}
 \big\langle A, \rho_{m+1} \big\rangle=\big\langle \mathcal{K}^* A, \rho_m  \big\rangle,
\end{equation}
where, $\mathcal{K}^*$ stands for the adjoint of $\mathcal{K}$. Namely,
\begin{equation}
    \mathcal{K}^*[A]:= \sum_{j} F_j^\dagger A F_j.
\end{equation}
\end{lemma}
\begin{proof}
    By the cyclic property of the trace operator, we have
    \begin{equation*}
        \begin{aligned}
            \tr(A\rho_{m+1}) = \sum_j \tr(AF_j\rho_m F_j^{\dag}) = \sum_j \tr(F_j^{\dag} AF_j\rho_m) = \tr(\mathcal{K}^*[A]\rho_m)
        \end{aligned}
    \end{equation*}
\end{proof}

In light of this Lemma, we have, 

\begin{theorem}
\label{theorem: 1}
When the density operator is simulated by the form of $\rho_{m+1} = \mathcal{K}[\rho_m] :=\sum_j F_j\rho_m F_j^{\dag}$ for each simulation time step $\delta t$, the explicit form of the derivative of the objective function $\phi(\bm \theta)$ can be written as 
\begin{equation}\label{deriv}
\partial_{\theta_{\alpha}}\phi(\bm \theta) = \frac{1}{N_O N_T} \sum_{k = 1}^{N_O}\sum_{n=1}^{N_T}\sum_{l=1}^{nL}\left(y_{k,n}-y_{k,n}^*\right) \Big\langle\partial_{\theta_\alpha}\mathcal{K}^*A^{(k)}_{nL-l}, \;\rho_{nL-l}\Big\rangle,
\end{equation} 
where $A^{(k)}_{nL-l}:=\big(\mathcal{K}^*\big)^{l-1}[A^{(k)}]$ can be viewed as a back-propagated operator.
\end{theorem}
\begin{proof}
    The detailed proof can be found in \cref{appendix: proof thm1}.
\end{proof}

By using \cref{theorem: 1}, the explicit expression of $\partial_{\theta_{\alpha}}\phi(\bm \theta)$ can be obtained once we know the derivative of $F_j$. For instance, for the first-order semi-implicit Euler method,
\begin{equation}
\label{eq:derivF1order}
\begin{aligned}
    \partial_{\theta} F_0 &= F_0\partial_{\theta}G F_0\delta t    \\
    \partial_{\theta}F_j &= \partial_{\theta}F_0V_j\sqrt{\delta t}+F_0\partial_{\theta}V_j\sqrt{\delta t}, j = 1,\cdots,N_V \\
\end{aligned}
\end{equation}
Similarly, for the second-order implicit approximation \eqref{2ndorder}, we have,
\begin{equation}
\label{eq:derivF2order}
    \begin{aligned}
        &\partial_{\theta}F_0 = \frac{1}{2}\delta t(I-\frac{1}{2}G\delta t)^{-1}\partial_{\theta}G(F_0+I)\\
        &\partial_{\theta}F_j = \frac{1}{2}\delta t(I-\frac{1}{2}G\delta t)^{-1}\partial_{\theta}GF_j+\sqrt{\delta t}(I-\frac{1}{2}G\delta t)^{-1}\partial_{\theta}V_j(I+\frac{1}{2}G\delta t) \\
        &+\frac{1}{2}\delta t^{3/2}(I-\frac{1}{2}G\delta t)^{-1}V_j\partial_{\theta}G,  \quad j=1,\cdots,N_V\\
        &\partial_{\theta}F_{j_1+N_Vj_2} = \frac{\delta t}{\sqrt{2}}(\partial_{\theta}V_{j_1}V_{j_2}+V_{j_1}\partial_{\theta}V_{j_2}),\quad j_1,j_2= 1,\cdots,N_V\\
    \end{aligned}
\end{equation}

{\color{black}
We now discuss the time complexity of the overall learning algorithm. 
Notice that simulating Lindblad dynamics for the time duration $[0,T]$, due to the second order accuracy,  incurs complexity that is proportional to the number of time steps, i.e., $LN_T = \mathcal{O}\left(\epsilon^{-1/2}T^{3/2}\right)$. Meanwhile, the quadratic convergence property of the Levenberg-Marquardt algorithm implies that the number of iteration steps is of the order $\mathcal{O}(\log \log \epsilon^{-1})$ such that the optimization error is within $\epsilon$. In addition, with direct calculation, we find that the time complexity of calculating the objective function value and its gradient is $\mathcal{O}(\epsilon^{-1/2}T^{3/2}{}N_V^2N_{MM})$ and  $\mathcal{O}(N_{\theta}N_O\epsilon^{-1/2}T^{3/2}N_V^2N_{MM})$, respectively, where $N_{MM}$ denotes the complexity associated with a matrix multiplication, e.g., $F_j \rho F_j^\dag$.  Thus the overall complexity is $ \Tilde{\mathcal{O}}(N_{\theta}N_O\epsilon^{-1/2}T^{3/2}N_V^2N_{MM})$ where $\Tilde{\mathcal{O}}$ ignores logarithmic factors. 

\begin{theorem}
 Given the data sets $\{y_{k,n}^*, k = 1,\cdots,N_O,n = 1,\cdots,N_T\}$ representing measurements with respect to observables $A^{(k)}$ and time instance $t_n$ until  $T=t_{N_T}$. Under the assumption in \cref{theorem: quadratic convergence} and given the observation data and precision $\epsilon$, the learning algorithm yields parameters $\bm \theta$ within  $\epsilon$ precision and it involves  $\Tilde{\mathcal{O}}(N_{\theta}N_O\epsilon^{-1/2}T^{3/2}N_V^2N_{MM})$ arithmetic operations. 
  For a specific $n$-qubit open quantum system described by Lindblad master equation, where the Hamiltonian $H$ is $k$-local and the dissipation part $\mathcal{L}$ only contains single qubit terms, 
    the overall complexity of the learning algorithm based on one and two-qubit Pauli observables  is
    \begin{equation*}
        \tilde{\mathcal{O}}\left(\epsilon^{-1/2}n^5T^{3/2}N_{MM}\right).
    \end{equation*}
\end{theorem}
Implicit in the final bound  in the above theorem is a $k$-dependent prefactor, which depends on the structure of the locality terms. The dependence of the bound on $n$
comes from the fact that for this specific quantum system, we have $N_O= \mathcal{O}(n^2),$ $N_\theta= \mathcal{O}({n}),$ and $N_V= \mathcal{O}(n)$. 
}

\subsection{Quantifying the Error in the Parameter Identification}\label{sec: error}

Similar to modern machine learning problems, 
there are mainly three sources of error in a quantum learning problem.  Specifically, as can be seen from  the objective function \cref{eq: phi}, 
\begin{enumerate}
    \item The estimated values, $\text{tr} \big( A^{(k)}  \rho(t_{n} ;\bm \theta) \big)$  are replaced with 
    $\text{tr} \big( A^{(k)}  \rho_{nL}(\bm \theta) \big)$ (Recall that $\Delta t = L \delta t $, and $\rho(t_n,\bm \theta) = \rho(\tau_{Ln},\bm \theta) \approx \rho_{Ln}(\bm \theta)$). This can be regarded as a function approximation error.  To include this error, we define,
    \begin{equation}\label{eq: phi1}
    \widehat{r}_{k,n}(\bm \theta) = \text{tr} \big( A^{(k)}  \rho_{nL}(\bm \theta) \big)  - y_{k,n}^*. 
    \end{equation}

    Using the second-order method \eqref{2ndorder}, $ \widehat{R}=(\widehat{r}_{k,n})_{k,n}$ is a $\mathcal{O}(\delta t^2)$ perturbation of the residual error in the original objective function \eqref{eq: phi}. 
    \item The data $y_{k,n}^*$, as indicated in \eqref{ykn}, come from measuring a set of observables at different time instances, which are subject to statistical error. To clarify this perspective,  we express the measurement values as random variables:
    \begin{equation}
        y_{k,n}^* \approx \hat{y}_{k,n} := \frac{1}{N_S} \sum_{\ell=1}^{N_S} a_{k,n,\ell}.
    \end{equation}
$N_S$ indicates the number of times the measurements are repeated.    The effect of this sampling error can be understood by considering the following objective function,
     \begin{equation}\label{eq: phi2}
    \widetilde{r}_{k,n}(\bm \theta) =   \text{tr} \big( A^{(k)}  \rho(t_n,\bm \theta)\big) -  \hat{y}_{k,n}. 
    \end{equation}
Based on Chebyshev's inequality,  $\widetilde{R}=(\widetilde{r}_{k,n})_{k,n}$ is a $\epsilon$ perturbation of residual error in the original objective function \eqref{eq: phi} with high probability if $N_S=\mathcal{O}(1/\epsilon^2)$. 

    \item The optimization problem \eqref{eq: minphi} may lead to a nonlinear system of equations and we may not be able to find the exact solution after a finite number of iterations. This is the optimization error.  

\end{enumerate}

The perturbation caused by the numerical approximation \eqref{2ndorder} is a deterministic perturbation, which from \eqref{ident}, causes a perturbation of the identified parameter. 
\begin{theorem}
    Under the assumption \eqref{ident}, 
    there exists a $\delta_0>0,$ such that for all $\delta t < \delta_0$,
    there is a minimizer $\hat{\bm \theta}$ of the least squares  problem defined by the residual vector in   \cref{eq: phi1}, such that 
    \begin{equation}
        \norm{ \hat{\bm \theta}-\bm \theta^* } \leq C \delta t^2. 
    \end{equation}
    for some constant $C$.
\end{theorem}

We now turn to the measurement error. Since the observables $A^{(k)}$'s are bounded, the measurement noise is bounded, and they can be regarded as a sub-Gaussian random variable. The Hoeffding inequality \cite{jin2019short} implies that, 
\begin{lemma}
  There exists $\sigma_{k,n}$, such that  the following inequality holds for every $\epsilon>0$
 \begin{equation}\label{subg-ci}
      \mathbb{P}\Big( \abs{ \hat{y}_{k,n} -  {y}_{k,n}^*  } > \epsilon  \Big)
\leq 2 e^{-\frac{N_S \epsilon^2}{ \sigma_{k,n}^2 }}.
 \end{equation}
\end{lemma}

Now by using a union bound, we can estimate the probability,
\[
 \mathbb{P}\Big(  \sqrt{\sum_{k,n}\abs{ \hat{y}_{k,n} -  {y}_{k,n}^*  }^2 } > \epsilon  \Big)
\leq 2 N_ON_T e^{-\frac{N_S \epsilon^2}{ N_O N_T \sigma^2 }}. 
\]
Here we have set $\sigma^2 = \frac{1}{N_O N_T}\sum_{k,n} \sigma_{k,n}^2.$

\begin{theorem}
    Under the same assumption as in the previous theorem, and further assume that the observation data $y_{k,n}$ are sampled,
    \begin{equation}
        N_S= \Omega\left( \frac{\sigma^2 N_ON_T}{\epsilon^2} \log \big(N_ON_T\big)\right),
    \end{equation}
    times, 
then       there is a minimizer $\widetilde{\bm \theta}$ of the nonlinear least squares with residual in \cref{eq: phi2}, such that 
    \begin{equation}
        \norm{ \widetilde{\bm \theta}-\bm \theta^* } \leq \epsilon. 
    \end{equation}
 with high probability.
\end{theorem}

{\color{black}
\subsection{Quantum/classical Hybrid Algorithms for Lindblad Simulation}

\begin{figure}[htp]
    \centering
    \includegraphics[scale=0.34]{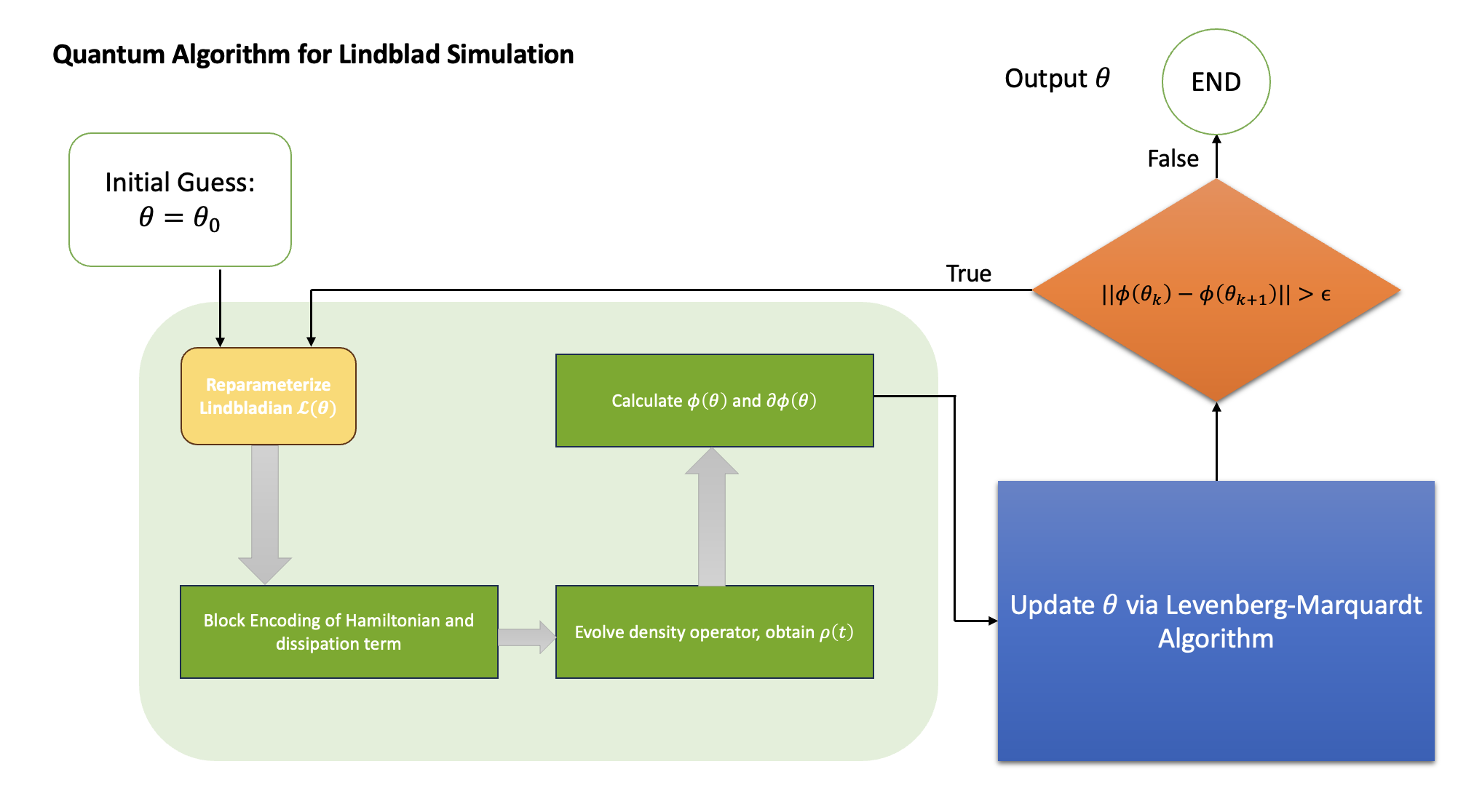}
    \caption{\color{black} A flowchart describing the parameter estimation algorithm using a quantum algorithm for Lindblad simulation. The simulation scheme is implemented by three steps: 1) Block encoding the Hamiltonian term and jump operators in the Lindbladian with coefficients depending on the input parameter $\bm \theta$; 2) Evolve the density operator $\rho(t)$ to time $t$; 3) Measure  observables $A^{(k)}$  with the density operator $\rho(t)$ and time $t_n$, and calculate the objective function value $\phi(\theta)$ and its derivatives $\partial\phi(\theta)$. 4) The optimization part is based on Levenberg-Marquardt algorithm on a classical device.}
    \label{fig:flowchart2}
\end{figure}

We have thus far discussed a classical approach to simulating the Lindblad equation, which is used to evaluate the objective function $\phi(\bm \theta)$ and its gradients as shown in \cref{deriv}. This method is suitable for quantum systems of moderate size, such as those that can be efficiently  treated using large-scale parallel algorithms. Alternatively, these dynamics can be simulated directly on quantum computers. Several such algorithms have already been developed \cite{KBG11,CL17,CW17,LW22}. Due to the many technical elements involved in these algorithms, we will sketch the overall procedure and leave the detailed presentation of these methods in separate works.  As illustrated in \cref{fig:flowchart2}, the overall procedure can be regarded as a quantum/classical hybrid algorithm, where some approximation of the parameter $\bm \theta$ is fed into the Lindbladian, from which the objective function   $\phi(\bm \theta)$ is estimated as an expectation. In addition, a gradient estimation algorithm, such as the improved Jordan's algorithm  \cite{gilyen2019optimizing}, can be used to estimate the gradient of  $\phi(\bm \theta)$. On the other hand, we run the LM algorithm on a classical device to provide an update of the parameter, which re-enters the quantum Lindblad simulation until it reaches convergence.

}

\section{Numerical Tests}\label{sec: num}
\red{In this section, we present several numerical results to test the effectiveness of our learning algorithm. For the test problem, we consider a quantum system of qubits with dynamics described by Lindblad \cref{eq: lindblad}. In particular, we assume that $H$ is linear combination of $k-$local operators in that each term is acting on at most $k$ qubits. In addition, the jump operators $V_j$ are assumed to be $1-$qubit Paulis.
Specifically, we choose the number of qubits $n=6$.  The initial states are fixed as the product state with all spins up. Meanwhile, the observables are chosen to be 1 and 2-local Pauli Strings. The performance of our learning algorithm will be quantified with relative and absolute error with respect to 2-norm. The source code is available at \cite{LearningLindblad}. }


\subsection{Phase Damping Model with Linear Parameter Dependence}
We \bl{investigate} an \bl{$n$}-spin system with dephasing and amplitude damping noise on every qubit \cite{temme2017error}. The dynamics are described by the QME in \cref{eq: lindblad}, with the parametric form of the Hamiltonian part given by, 
\begin{equation}\label{par:1H}
    H = \sum_{j=1}^n H_j, \quad  H_j = \sum_{\alpha = 1}^3 e_{j,\alpha}\sigma_j^{\alpha}+\sum_{\alpha,\beta = 1}^3 c_{j,\alpha, \beta}\sigma_j^{\alpha}\sigma_{j+1}^{\beta},\quad j = 1,\cdots,n.
\end{equation}
 $\sigma_j^{\alpha}$ is the Pauli matrix  $\sigma^{\alpha}$  ($\alpha=x,y,z$) applied to the $j$th qubit, and \bl{the condition} $c_{n,\alpha,\beta} = 0$ \bl{models} open boundaries.

Meanwhile, the dissipation term  is parameterized as follows,
\begin{equation}\label{par:1D}
    \mathcal{L}_D\rho = \lambda_1\sum_{j = 1}^n\mathcal{T}[\sigma_j^-] \rho+\lambda_2\sum_{j=1}^n\mathcal{T}[\sigma_j^z]\rho,
\end{equation}
where the operator $\mathcal{T}[V]$ is defined as
\begin{equation}
    \mathcal{T}[V] \rho = V\rho V^{\dag}-\frac{1}{2}\{V^{\dag}V,\rho\}.
\end{equation}
\bl{The operators $\sigma_j^{\pm} := 2^{-1}(\sigma_j^x\pm i \sigma_j^y)$ are introduced.} 

\bl{The parameters are expressed as a vector $\bm \theta$, consisting of Hamiltonian coefficients $\bm \theta_H = \{e_{j,\alpha},c_{j,\alpha,\beta}\}$ and dissipative parameters $\bm \theta_D = (\lambda_1, \lambda_2)$, with 65 unknown parameters in total. To initialize the learning algorithms, }these parameters will be generated from a Gaussian distribution, and then held fixed as the exact values of the parameters. 
\medskip

We first test the accuracy of the simulation methods in \cref{sec: simulation}. 
As a reference, we conducted direct simulations of the test model in \cref{par:1H,par:1D} by using very small step size: $\delta t = 10^{-4}$. \bl{The measurement with 1-qubit local observable $A:= \sigma_2^y$, will serve as the benchmark solution and is denoted by $y_\text{exact}(t)$. }
\cref{fig:exp19Simu}  \bl{displays} the \bl{evolution} of this observable in the interval $t\in [0,10].$ \bl{It contrasts the "exact" solution with results from} the semi-implicit Euler method \eqref{semi-Euler} and the second-order implicit method \eqref{2ndorder} using step size $\delta t= 10^{-2}$,  \bl{referred to as} $y_1(t),$ and $y_2(t)$ respectively. One can observe that the semi-implicit Euler's method yields very good accuracy in the initial period $[0,0.5]$, but only a qualitatively correct solution in the transient state $[0.5,4]$. In addition, it tends to smear out the solution in the long run, e.g., in the final period $[6,10].$ But the second-order implicit method offers much better accuracy in the entire time duration.
Due to the high accuracy, the measurement data in  the following numerical experiments
will be generated using the second-order implicit method with $\delta t=10^{-2}$.

\begin{figure}[thp]
    \centering
    \includegraphics[scale=0.11]{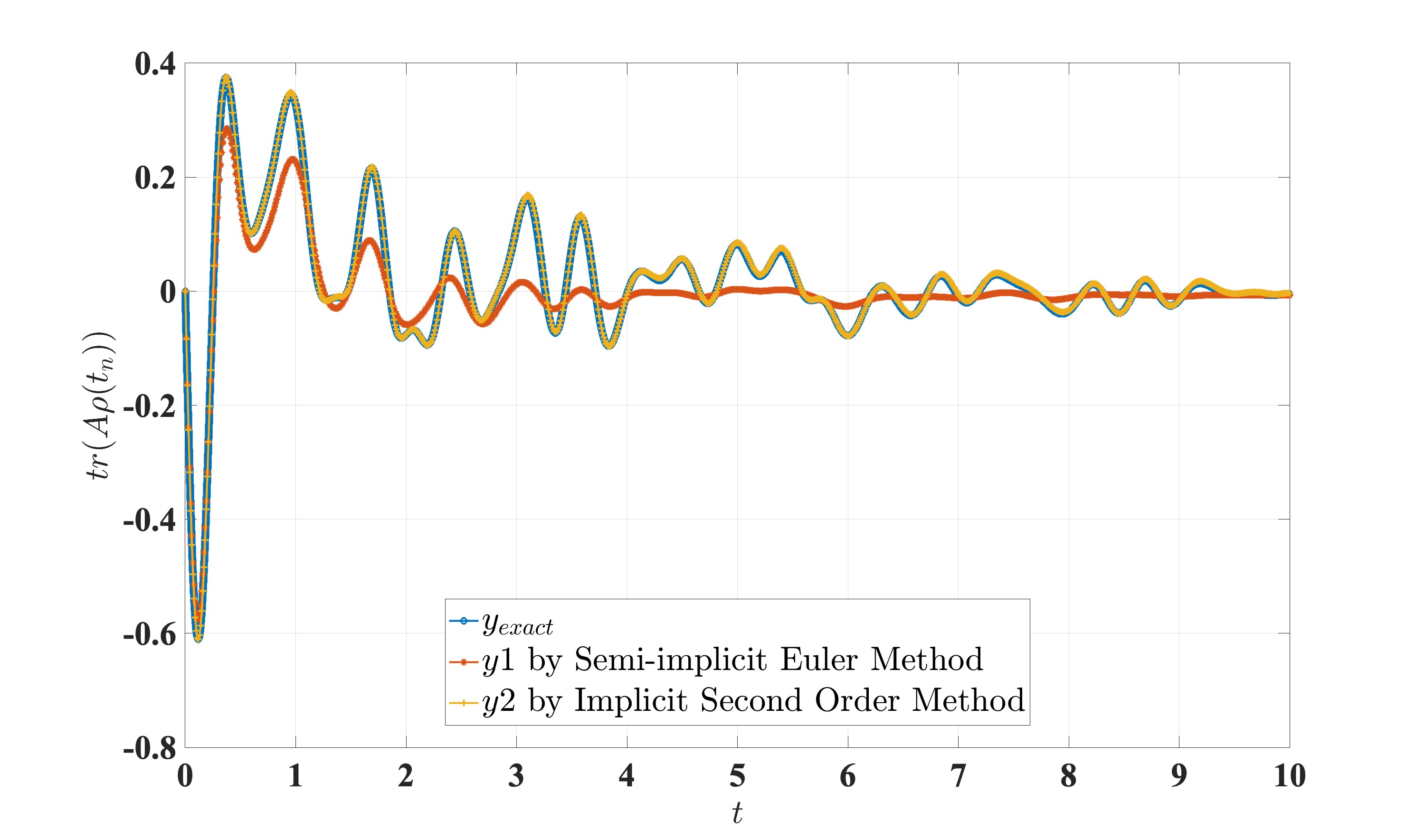}
    \caption{Comparison of the accuracy for the simulation methods in \cref{sec: simulation}, \bl{applied to a} $6$\bl{-spin} \bl{chain} with dephasing and amplitude damping noise. The density operator $\rho$ is measured with 1-qubit local observables $A := \sigma_2^y$ and the measurement \bl{at each time point denoted by} $y_n = \tr(A\rho(t_n))$.
    The blue solid line represents the "exact" solution generated  with a very small step size.}
    \label{fig:exp19Simu}
\end{figure}

\medskip

We \bl{ now employ the Lindbladian simulation scheme to implement the optimization algorithms detailed in \cref{sec: optm}}.
\bl{Observations were made at time points} $t_n=0.1,0.2,\cdots, 1$ with measurement interval $\Delta t=0.1$ \bl{and in the learning algorithm, the underlying Lindblad dynamics were simulated with a smaller step $\delta t = 0.01$ ($L=10$)}. \bl{The performance of the  algorithms was tested on the basis of all 1-local observables or all 1 and 2-local observables, generated by Pauli matrices. Numerical results in \cref{fig:exp30.3&4} include the objective function value and the relative error $\norm{\bm \theta -\bm \theta^*}/\norm{\bm \theta^*}$. Notably, the Levenberg-Marquardt algorithm demonstrated rapid convergence. Furthermore, with all 1 and 2-local observable, the optimization yields slightly faster parameter identification. It demonstrates that an increased number of observables enhances identification.}

\begin{figure}[htp]
    \centering
    \includegraphics[scale=0.11]{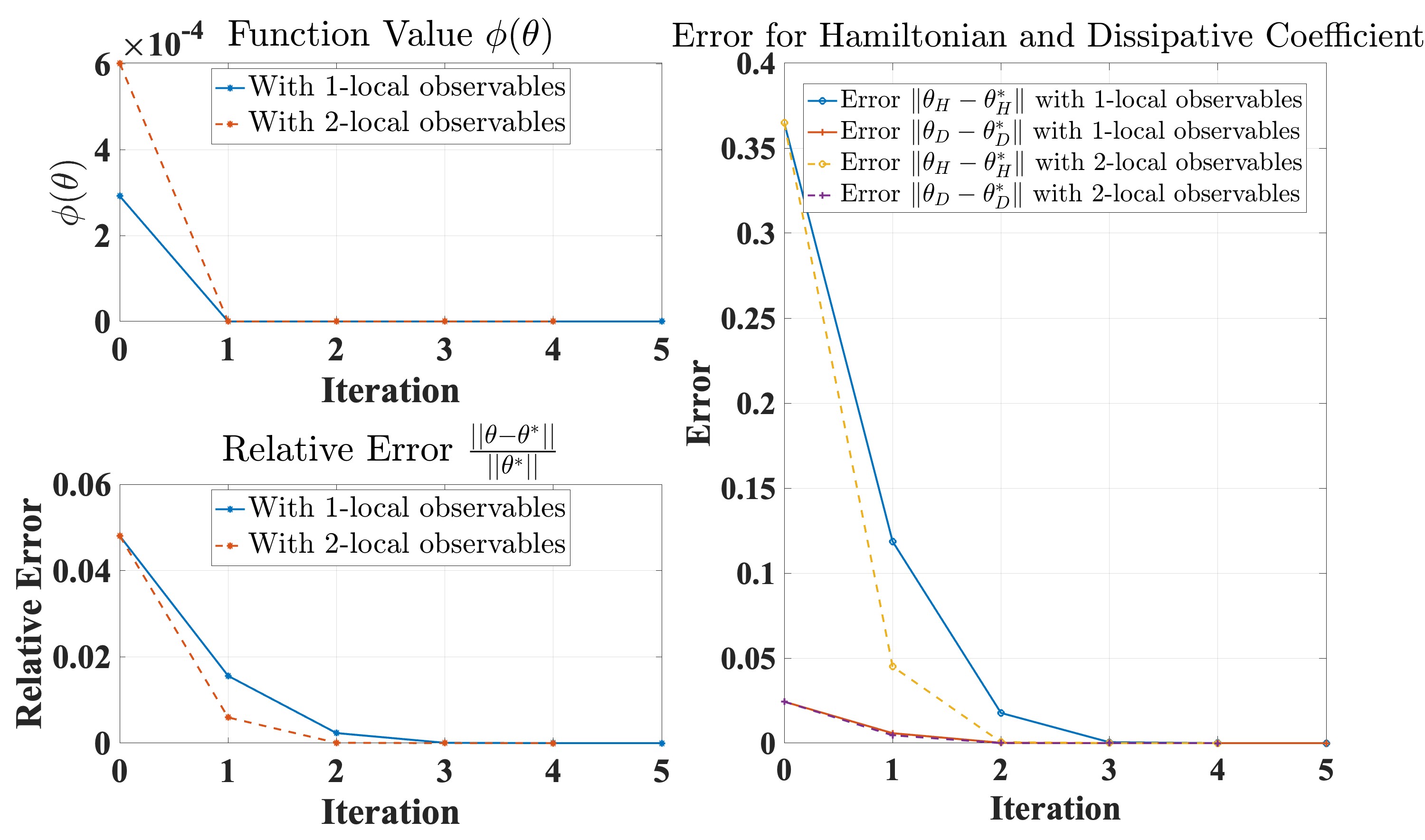}
    \caption{Reconstruction of the Lindbladians in \cref{par:1H,par:1D} by  solving the optimization problem \eqref{eq: phi}. \bl{The dataset consists of 1-qubit local observables ($N_O=19$) and both 1 and 2-local observables($N_O=154$) at time $t_n=0.1,0.2,\cdots, 1$. The numerical simulation time interval is smaller as $\delta t=0.01$.} The initial parameter  $\bm \theta^{(0)}$ is randomly selected with initial error $\| \bm\theta^{(0)}- \bm \theta^*\| = 0.3658$.
}
    \label{fig:exp30.3&4}
\end{figure}

\bl{Extending our analysis, we retained 1-local observables but increased measurement frequency, setting measurement times to $\tau = t =0.01,0.02,\cdots,1$ ($\Delta t=\delta t=0.01$). \cref{fig:exp28} indicates that higher measurement frequency  slightly enhances parameter identification.} The efficiency of the optimization can again be attributed to the rapid convergence of the LM algorithm. 
\begin{figure}
    \centering
    \includegraphics[scale=0.11]{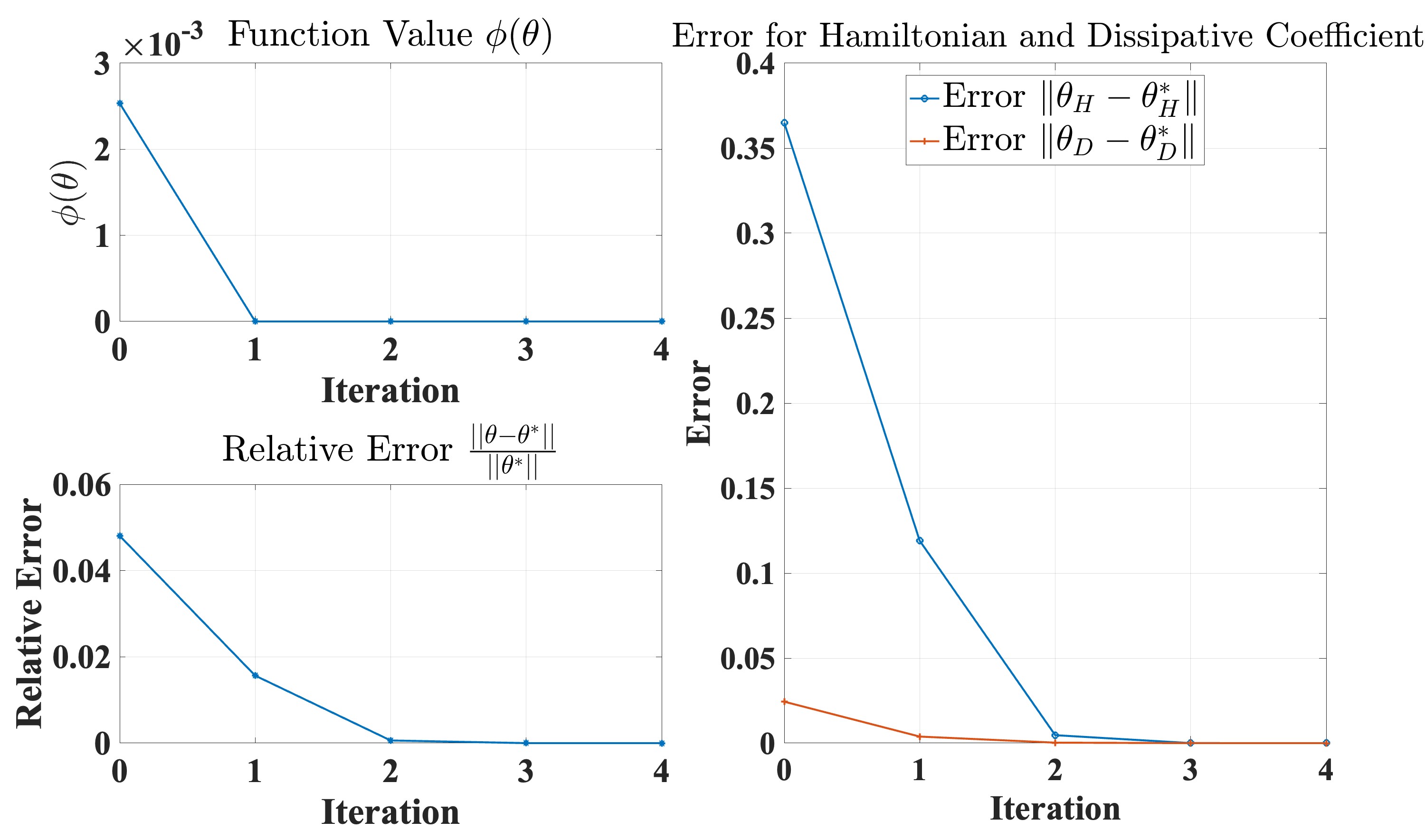}
    \caption{Reconstruction of the Lindbladians in \cref{par:1H,par:1D} by  solving the optimization problem \eqref{eq: phi}. \bl{The dataset consists of  1-qubit local observables ($N_O=19$) at time $t = 0.01,0.02,\cdots,1$($N_T = 100$). The numerical simulation time interval is the same as $\delta t=0.01$.}
    The initial parameter  $\bm\theta^{(0)}$ is the same as \bl{\cref{fig:exp30.3&4}.}
}
    \label{fig:exp28}
\end{figure}

\bl{Inspired} by the study in \cite{bairey2020learning}, where the Lindbladians are learned \bl{from} steady \bl{states}, \bl{our last experiment uses} observation times at a later stage $t = 4.1,4.2,\cdots,5$. A separate numerical test verified that this is when the system is beginning to saturate (see \cref{{fig:exp19Simu}} as well). \bl{As shown in} \cref{fig:exp30.6.1}, the optimization \bl{identified} all the parameters. \bl{Despite the non-monotonic convergence of $\bm \theta_H$, the last few iterations exhibited rapid convergence.}

\begin{figure}[htp]
    \centering
    {\includegraphics[scale=0.11]{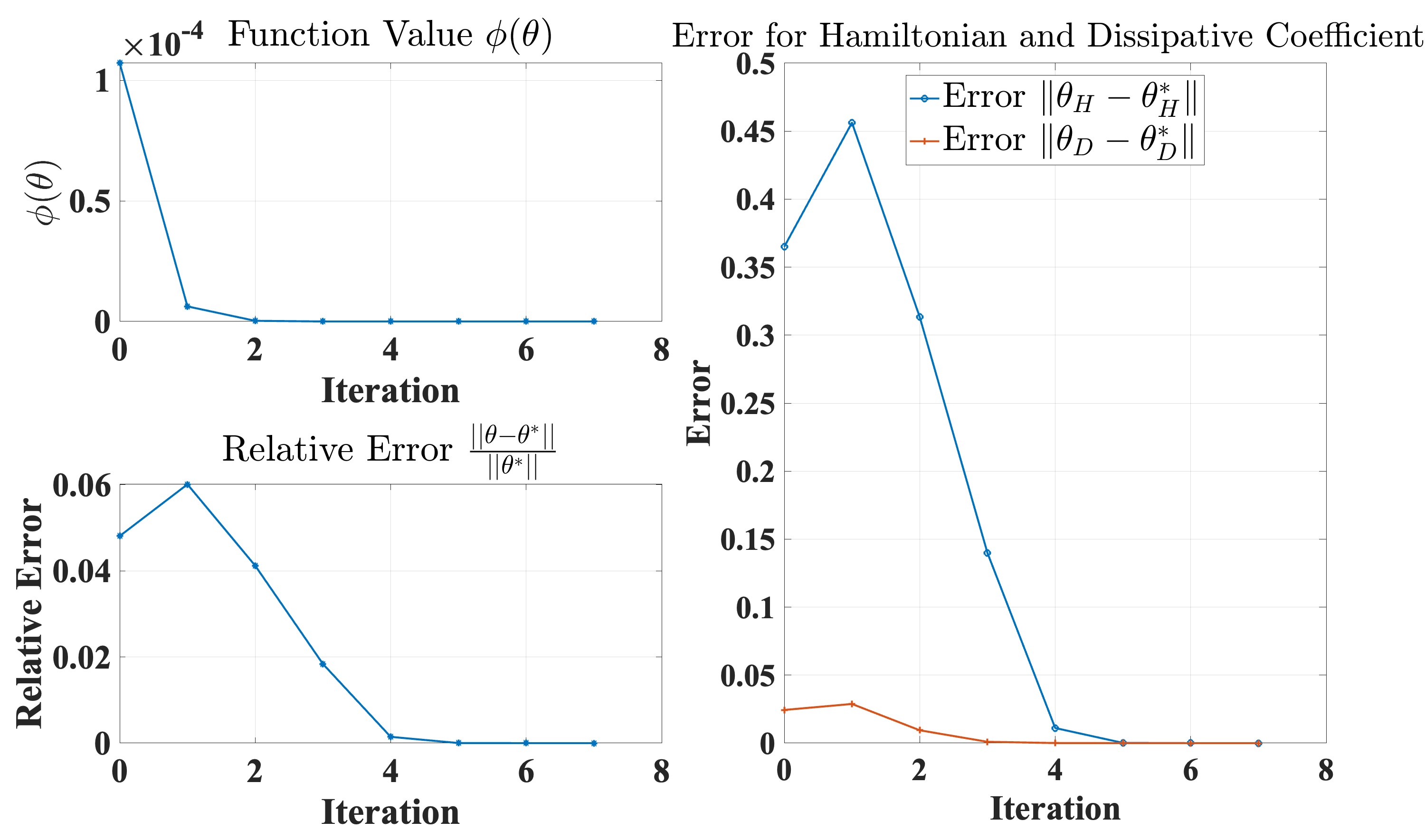}} 
    \caption{ Reconstruction of the Lindbladians in \cref{par:1H,par:1D} by  solving the optimization problem \eqref{eq: phi}. \bl{The dataset consists of  1-qubit local observable s($N_O=19$) at time $t = 4.1,4.2,\cdots,5$, ($N_T = 10$). The numerical simulation time interval is smaller as $\delta t=0.01$.}
    The initial parameter $\bm\theta^{(0)}$ is the same as \bl{\cref{fig:exp30.3&4}.}\label{fig:exp30.6.1}
    }
\end{figure}

\subsection{A Nonlinear Parametric Model}\label{sec: nonlinear}
\bl{
The model  in the previous section involves dissipation terms that are linear with respect to the parameters. In this section, we assume that each jump operator in the dissipation term has a linear parametric form, effectively leading to a nonlinear parametric model. 
This coincides with the example in Bairey et al. \cite{bairey2020learning}. Specifically, the Hamiltonian is of the same linear form as in the previous section while} the jump operators are expanded as a linear combination of local Pauli matrices \bl{with complex coefficients},
\begin{equation}
\label{par:2D}
 V_j = \sum_{\alpha=1}^3 d_{j,\alpha}\sigma_j^{\alpha} = \sum_{\alpha = 1}^3 d_{j,\alpha}^{(1)}\sigma_j^{\alpha}+id_{j,\alpha}^{(2)} \sigma_j^\alpha,\quad j = 1,\cdots, N_V.
\end{equation}
\bl{Similarly, the parameters are randomly generated by the Gaussian distribution},  
\begin{equation}\label{par2d}
    e_{j,\alpha},c_{j,\alpha,\beta}\sim\mathcal{N}(0,1),\quad\text{Re}(d_{j,\alpha}), \text{Im}(d_{j,\alpha})\sim\mathcal{N}(0,\frac{1}{2}).
\end{equation}
These parameters will then be fixed and considered to be exact. Subsequently, they are used to generate the observation data $y_{k,n}^*$ in \eqref{ykn}.  We have noticed that the parametric form in \cref{par2d} is not unique due to the simultaneous appearance of $V_j$ and $V_j^\dagger$ in the Lindbladian. 
To eliminate the redundancy from the global phase of $d_{j}$ for each $j$, we set \bl{the imaginary part of $d_{j,1}$ to zeros}.
The parameter $\bm \theta $ is composed by Hamiltonian part $\bm \theta_H = \{e_{j,\alpha}, c_{j,\alpha,\beta}\}$ as in \cref{par:1H} \bl{and the dissipation part} $\bm \theta_D = \{d_{j,\alpha}^{(1)},d_{j,\alpha}^{(2)}\}$. 

The optimization problem \eqref{eq: phi} with values $y_{k,n}^*$ \bl{involved both} 
 1-local and 2-local observables at time $t = 0.1,0.2,\cdots,1$ \bl{($\Delta t=0.1$)}. \bl{The simulation time interval is much smaller: $\delta t = 0.01$. These experiments (results shown in \cref{fig:exp1923}) indicate that effective parameter identification could be achieved even with the nonlinear parametric model. Moreover, faster convergence was observed when employing all 1 and 2-local observables, which could be possibly caused by a larger constant in the condition \eqref{ident}.}
To further illustrate the effect of the choice of the observables, we choose 1-local observables, by only keeping $\sigma_j^x$ and $\sigma_j^y$'s for each spin, while other configurations of the test remain the same. \bl{This modification led to a noticeable increase in the number of iterations to reach a plateau, as demonstrated in \cref{fig:exp23.2}. The large optimization error suggests that the limited set of 1-local observables is insufficient for the Lindbladian learning problem.}

\begin{figure}[hbtp]
    \begin{center}
        \includegraphics[scale=0.11]{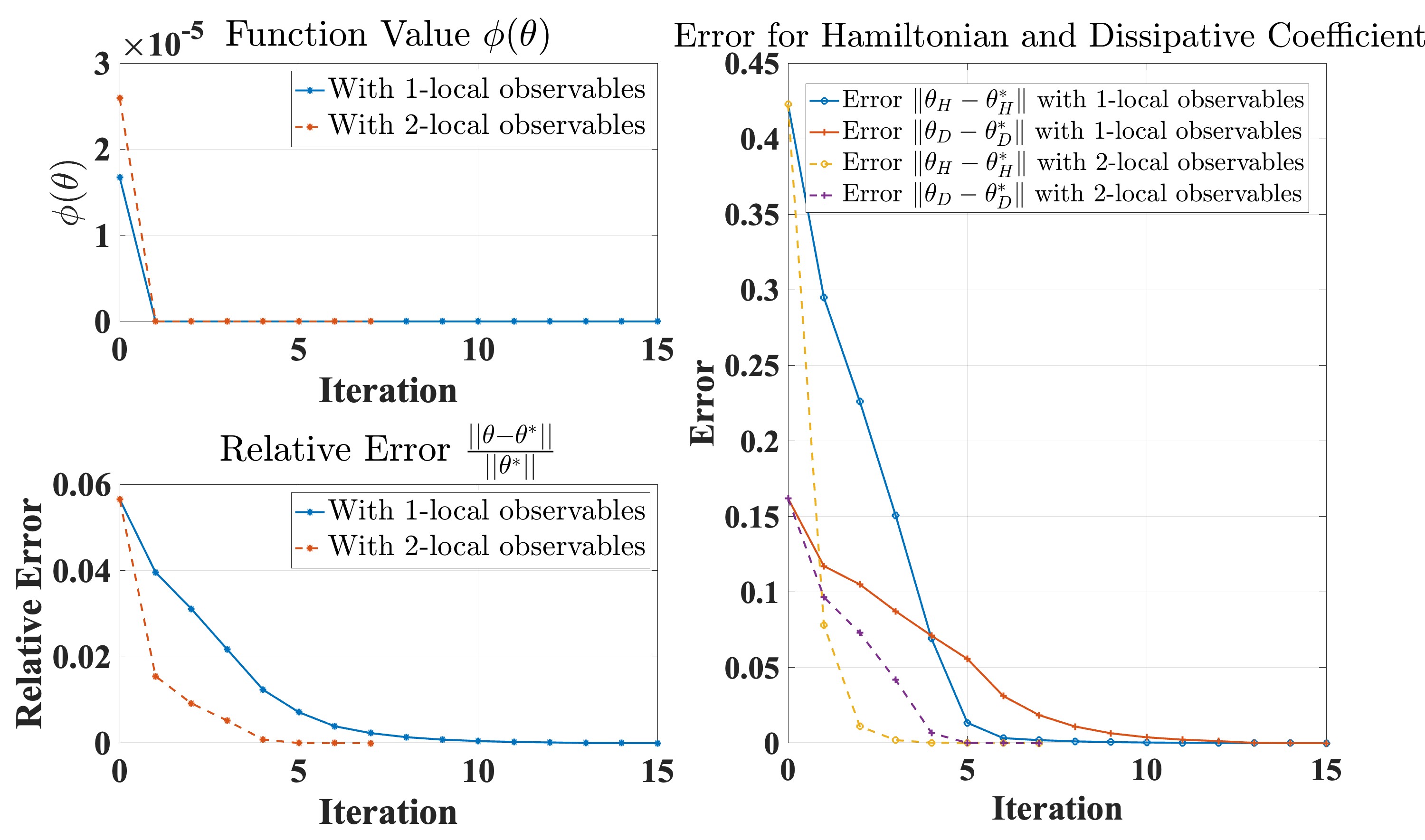}
    \end{center}
    \caption{Reconstruction of the Lindbladians in \cref{par:1H,par:2D} by  solving the optimization problem \eqref{eq: phi}. \bl{The dataset consists of  1-qubit local observables ($N_O=19$) and both 1 and 2-local observables($N_O=154$) at time $t_n=0.1,0.2,\cdots, 1$. The numerical simulation time interval is smaller as $\delta t=0.01$.}
    The initial parameter  $\bm \theta^{(0)}$ is randomly selected with initial error $\|\bm \theta^{(0)}- \bm \theta^*\| = 0.4529$. 
    }
    \label{fig:exp1923}
\end{figure}

\begin{figure}[hbtp]
    \centering
    \includegraphics[scale=0.11]{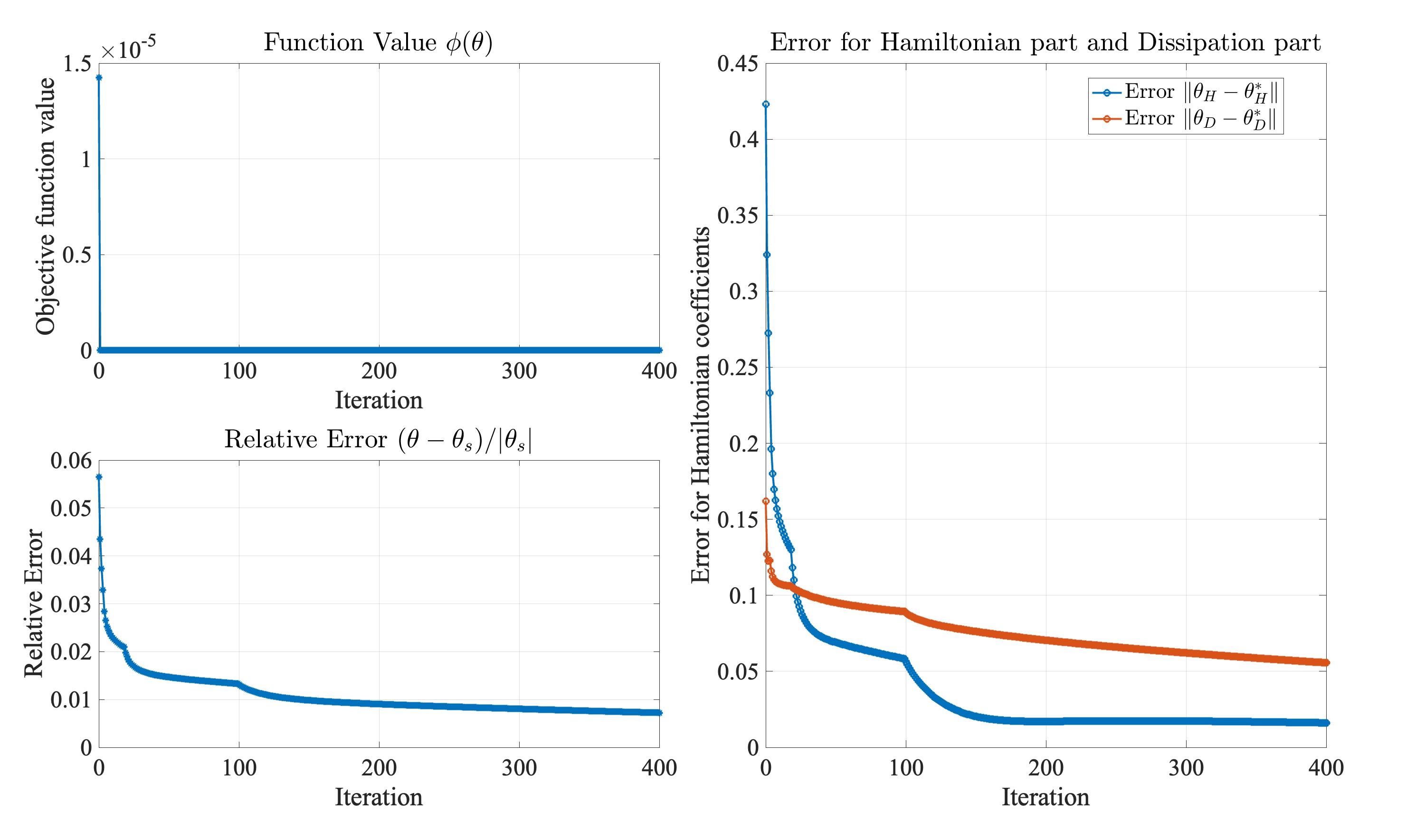}
    \caption{ Reconstruction of the Lindbladians in \cref{par:1H,par:2D} by  solving the optimization problem \eqref{eq: phi}. \bl{The dataset consists of only 1-qubit local observables $\sigma^x_j,\sigma^y_j$ ($N_O=12$) at time $t_n=0.1,0.2,\cdots, 1$ ($N_T = 10$). The numerical simulation time interval is smaller as $\delta t=0.01$.}
    The initial parameter   $\bm\theta^{(0)}$ is the same as the previous test. 
    }
    \label{fig:exp23.2}
\end{figure}

\bl{ So far we only considered a small initial guess to show the second-order convergence property of Levenberg-Marquardt algorithm. To illustrate the robustness of our algorithm, 
we implemented three additional numerical experiments with initial guesses chosen such that the difference between the initial parameter $\bm \theta^{(0)}$ and the minimizer $\bm\theta^*$ is $\mathcal{O}(1)$ for both linear and nonlinear cases. We repeat the numerical experiment in \cref{fig:exp28}, but using a random initial guess with a much larger error $\|\bm\theta^{(0)}-\bm\theta^*\| = 2.3650$. The convergence is within 10 iterations as shown in \cref{fig:exp28.1 further3}. 
\begin{figure}
    \centering
    \includegraphics[scale=0.11]{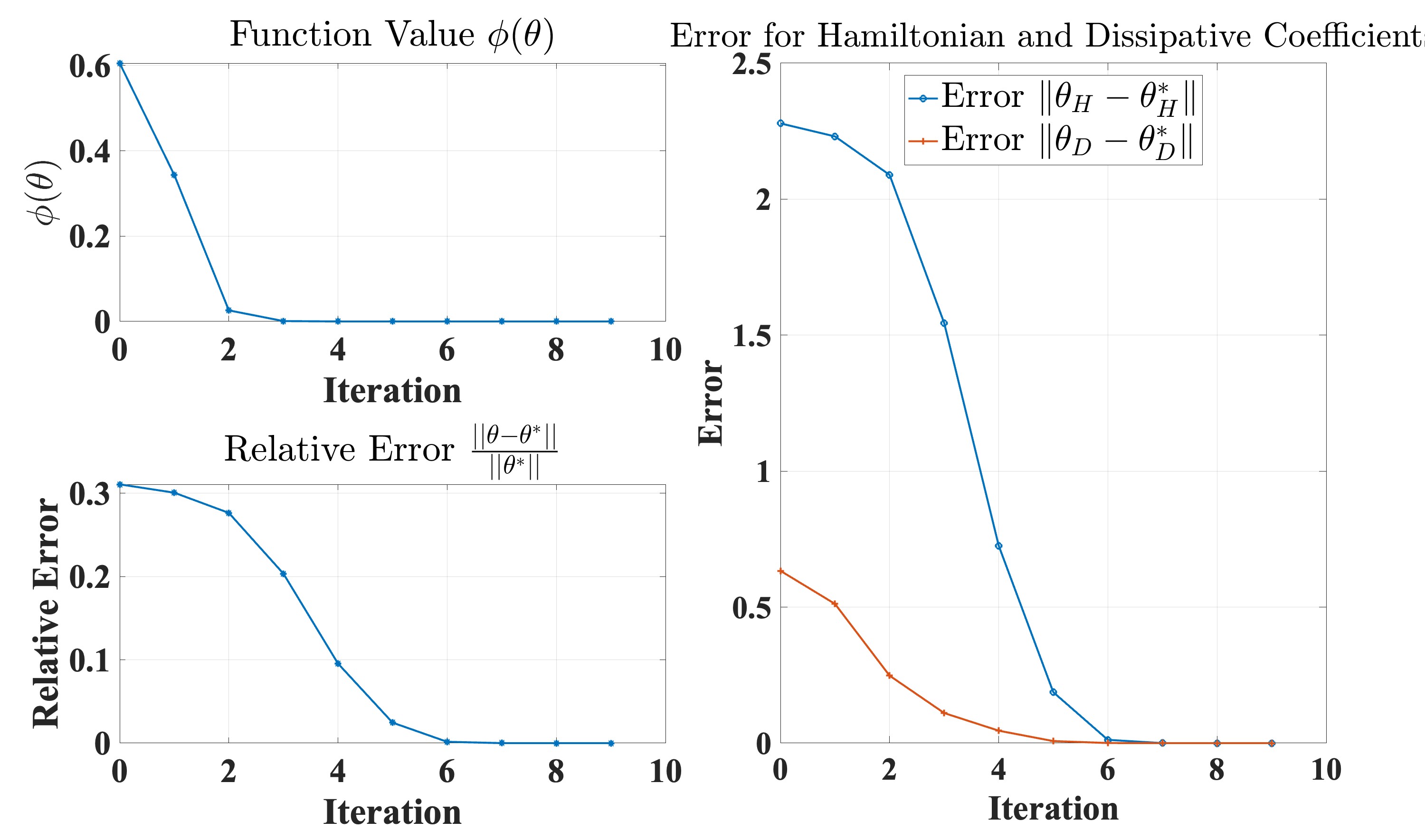}
    \caption{\bl{Reconstruction of the Lindbladians of the linear parameter case. The dataset includes  1-qubit local observables ($N_O=19$) at time $t = 0.01,0.02,\cdots,1$ ($N_T = 100$). The numerical simulation time interval is the same as $\delta t=0.01$. Same setting as \cref{fig:exp28}. 
    The initial guess is larger,  $\|\bm\theta^{(0)}-\bm\theta^*\| = 2.3650$.}}
    \label{fig:exp28.1 further3}
\end{figure}
Similarly, we repeat the experiment in \cref{fig:exp1923} and the initial guess has error $\|\bm\theta^{(0)}-\bm\theta^*\| = 2.0359$. Our algorithm still converged quite rapidly as depicted in \cref{fig:exp25.3 further3}.
\begin{figure}
    \centering
    \includegraphics[scale=0.11]{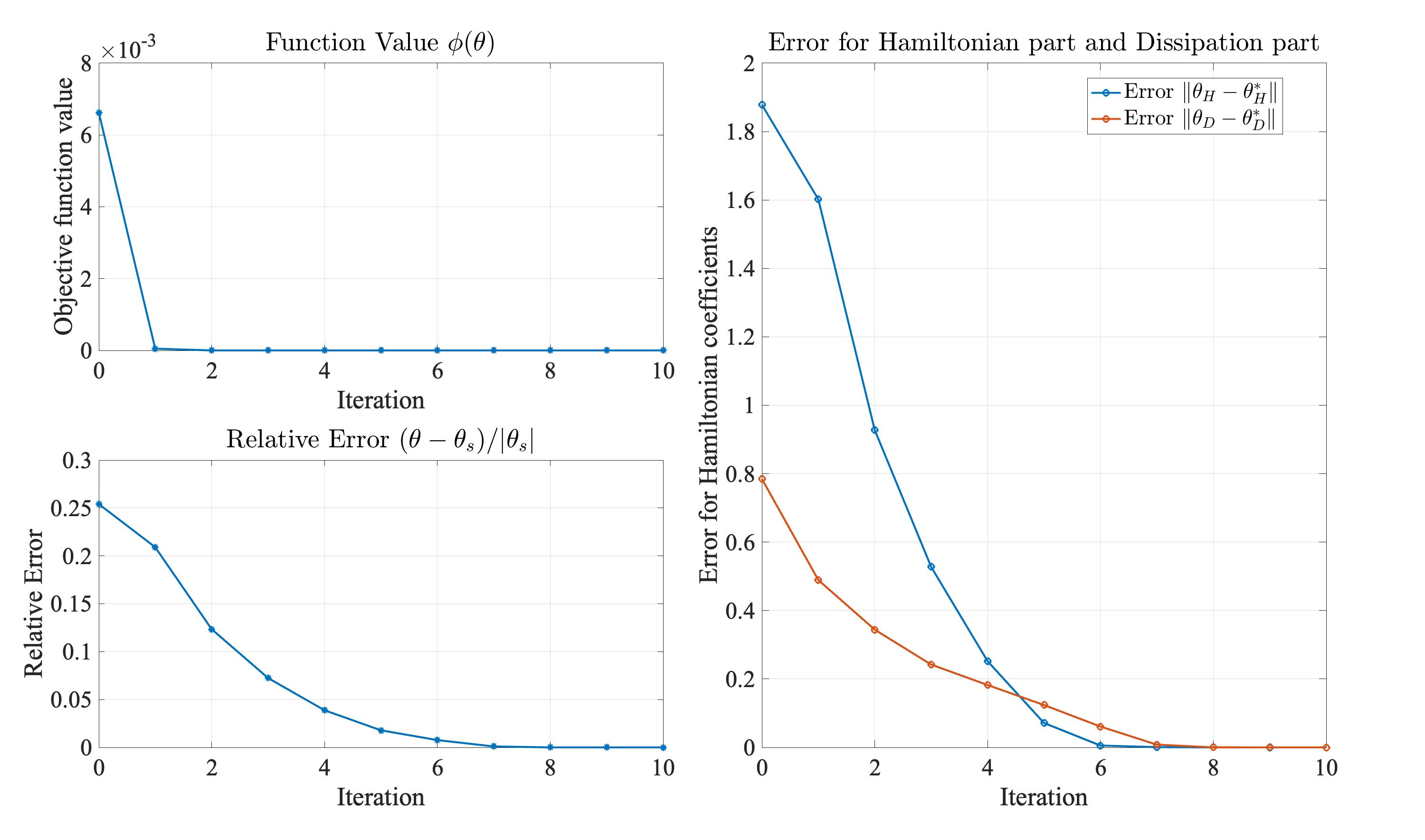}
    \caption{\bl{Reconstruction of the Lindbladians of the nonlinear parameter case. The dataset consists of  1-qubit local observables ($N_O=19$) at time $t = 0.1,0.2,\cdots,1$ ($N_T = 10$). The numerical simulation time interval is the same as $\delta t=0.01$. Same setting as \cref{fig:exp1923}. 
    The initial guess is larger,  $\|\bm\theta^{(0)}-\bm\theta^*\| = 2.0359$.}}
    \label{fig:exp25.3 further3}
\end{figure}
We now choose $\bm \theta^{(0)}$ with a much larger initial error $\|\bm \theta^{(0)}-\bm \theta^*\| = 7.8569$ in the  experiment in \cref{fig:exp28}. As we can observe from \cref{fig:exp28.1 further2}, the algorithm still exhibits convergence but the parameter $\bm\theta$ has settled to another minimum. Nevertheless, the objective function is almost zero, indicating an excellent fit to the observation data. 
In summary, our algorithm demonstrates robust performance even with a considerably larger initial guess.
However, the algorithm might end up with another global minimum $\bm\theta$.
\begin{figure}
    \centering
    \includegraphics[scale=0.11]{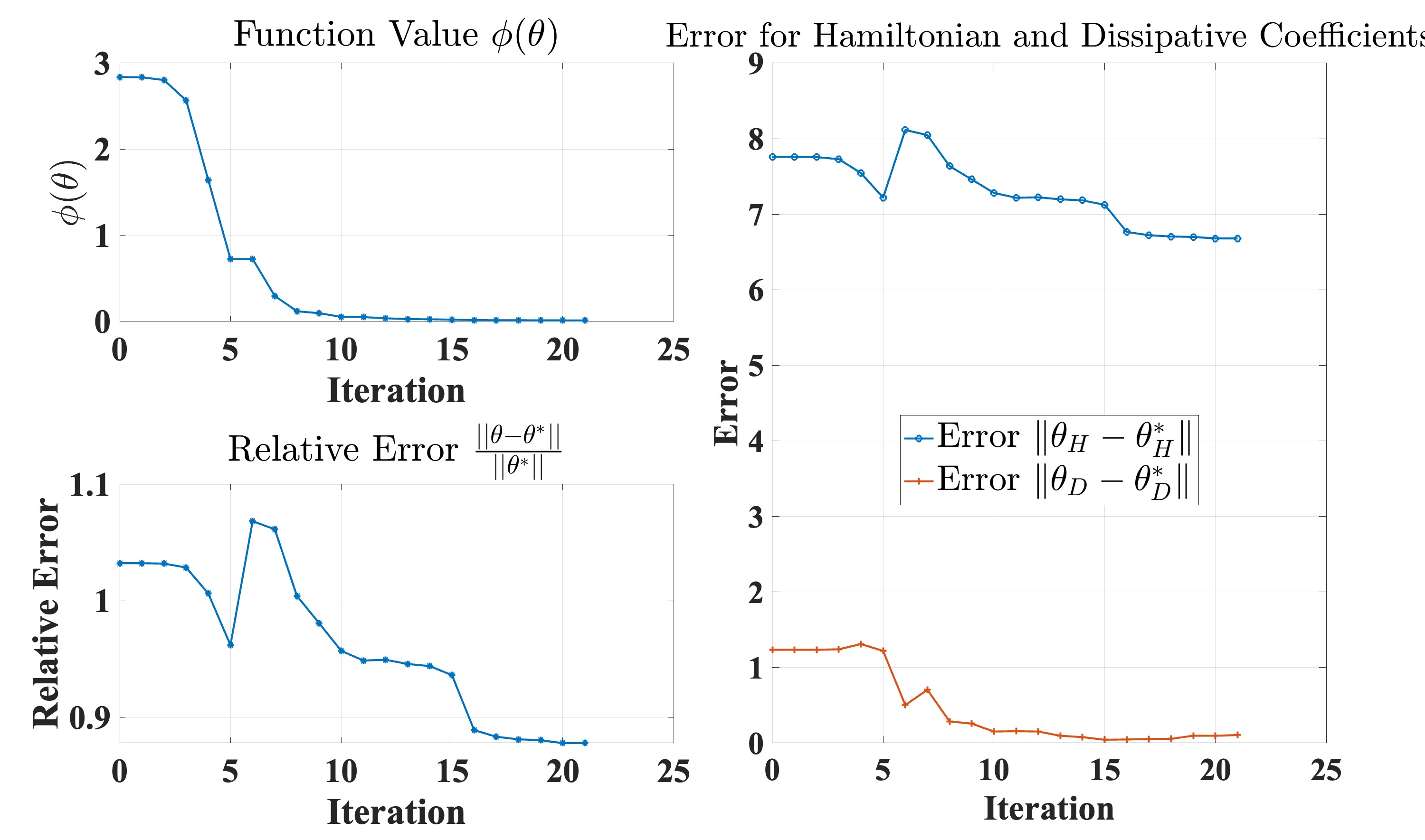}
    \caption{\bl{Reconstruction of the Lindbladians of the linear parameter case. The dataset consists of 1-qubit local observables ($N_O=19$) at time $t = 0.01,0.02,\cdots,1$ ($N_T = 100$). The numerical simulation time interval is the same as $\delta t=0.01$. Same setting as \cref{fig:exp28}. 
    The initial guess is larger,  $\|\bm\theta^{(0)}-\bm\theta^*\| = 7.8569$.}}
    \label{fig:exp28.1 further2}
\end{figure}}

\section{Summary and discussions} 
In this paper, we presented an algorithm to infer parameters in an open quantum system, specifically, in a Lindblad equation. Rather than working with the observation times, we introduce smaller time scales where the solution is obtained by direct numerical simulations. There are two advantages. On one hand, since the objective function is formulated based on trajectories, we no longer need the observables that correspond to the Lindbladian terms that appeared in the approach in \cite{bairey2020learning,franca2024efficient}. As a result, the number of observables can be much less than the number of jump operators. On the other hand, the algorithm enables flexible control of the accuracy by using smaller step sizes in the numerical simulation.

\section*{Acknowledgements}
We thank Dr. Di Fang for providing some important references.  This work is supported by NSF Grant DMS-1819011 and DMS-2111221. 

\bibliographystyle{quantum}




\appendix
\section{Appendix}
\subsection{Proof of \cref{theorem: 1}}
\label{appendix: proof thm1}
\begin{proof}
The derivative of the objective function $\phi$ is 
\begin{equation}
    \partial_{\theta_{\alpha}}\phi(\bm\theta)=\frac{1}{N_ON_T} \sum_{k = 1}^{N_O}\sum_{n=1}^{N_T}(y_{k,n}-y_{k,n}^*) \big\langle A^{(k)},\partial_{\theta_{\alpha}}\rho(\tau_{nL};\bm\theta)\big\rangle
\end{equation}
Let $\chi_{nL}^{\alpha} = \partial_{\theta_{\alpha}}\rho(\tau_{nL};\bm\theta)$ and $\rho_{nL} := \rho(\tau_{nL};\theta)$. The iteration formula for the density operator is of the Kraus form $\rho_{m+1} = \mathcal{K}[\rho_{m}] := \sum_j F_j\rho_{m}  F_j^{\dag}, m = 1,\cdots,nL$. It reveals $\chi_{m+1}^{\alpha} = \mathcal{K}[\chi_m^{\alpha}]+\partial_{\theta_{\alpha}}\mathcal{K}[\rho_m]$, then
\begin{equation}
    \begin{aligned}
\big\langle A^{(k)},\chi_{nL}^{\alpha}\big\rangle&=\big\langle A^{(k)},\mathcal{K}\chi_{nL-1}^{\alpha})+\big\langle A^{(k)},\partial_{\theta_{\alpha}}\mathcal{K}\rho_{nL-1}\big\rangle\\
\big\langle A^{(k)},\mathcal{K}\chi_{nL-1}^{\alpha}\big\rangle& = \big\langle A^{(k)},\mathcal{K}^2\chi_{nL-2}^{\alpha}\big\rangle+\big\langle A^{(k)},\mathcal{K}(\partial_{\theta_{\alpha}}\mathcal{K})\rho_{nL-2}\big\rangle\\
\vdots & \\
\Rightarrow \big\langle A^{(k)},\chi_{nL}^{\alpha}\big\rangle&=\big\langle A^{(k)},\mathcal{K}^{nL}\chi_0^{\alpha}\big\rangle+\sum_{l=1}^{nL}\big\langle A^{(k)},\mathcal{K}^{l-1}(\partial_{\theta_{\alpha}}\mathcal{K})\rho_{nL-l}\big\rangle\\
&=\sum_{l=1}^{nL}\big\langle A^{(k)},\mathcal{K}^{l-1}(\partial_{\theta_{\alpha}}\mathcal{K})\rho_{nL-l}\big\rangle\\
    \end{aligned}
\end{equation}
as $\chi_0^{\alpha} = 0$. By taking the adjoint operator,
\begin{equation}
    \begin{aligned}
        \partial_{\theta_{\alpha}}\phi(\bm\theta)
        &= \frac{1}{N_ON_T}\sum_{k,n}(y_{k,n}-y_{k,n}^*)\sum_{l=1}^{nL}\big\langle(\partial_{\theta_{\alpha}}\mathcal{K})^*(\mathcal{K}^*)^{l-1}A^{(k)},\rho_{nL-l}\big\rangle\\
        & = \frac{1}{N_ON_T}\sum_{k,n}(y_{k,n}-y_{k,n}^*)\sum_{l=1}^{nL}\big\langle(\partial_{\theta_{\alpha}}\mathcal{K})^*A^{(k)}_{nL-l},\rho_{nL-l}\big\rangle\\
    \end{aligned}
\end{equation}
where $A^{(k)}_{nL-1}: = (\mathcal{K}^*)^{l-1}[A^{(k)}]$ is a back-propagated operator and the adjoint Kraus form
\begin{equation}
   (\partial_{\theta_\alpha}\mathcal{K})^*(\rho) = \sum_j (\partial_{\theta_\alpha} F_j)^{\dag}\rho F_j +F_j^{\dag}\rho (\partial_{\theta_\alpha}F_j),\quad \mathcal{K}^*(\rho) = \sum_{j} F_j^{\dag}\rho F_j
\end{equation}
for real parameters $\bm \theta$.
\end{proof}

\subsection{Second Order Implicit Approximation Algorithm}
\label{appendix: 2.0 scheme}

Similar to Lemma 1, we start with the unraveled SDE, 
\begin{equation}
    \mathrm{d}\ket{\psi} = (-iH_S\ket{\psi}-\frac{1}{2}\sum_{j=1}^{N_V} V_j^{\dag}V_j\ket{\psi} )\mathrm{d}t +\sum_{j=1}^{N_V} V_j\ket{\psi} \mathrm{d}W_{j;t}=G\ket{\psi} \mathrm{d}t +\sum_{j=1}^{N_V} V_j\ket{\psi} \mathrm{d}W_{j;t}
\end{equation}
According to the second order implicit weak scheme\cite[Chapter 15]{kloeden2011numerical}, we can write down a time-marching scheme for the wave function,
\begin{equation}\label{sde2nd}
\begin{aligned}
    \ket{\psi_{m+1}} &= \ket{\psi_m}+\frac{1}{2}(G\ket{\psi_{m+1}}+G\ket{\psi_m})\delta t+\sum_{j=1}^{N_V} V_j\ket{\psi_m}\delta  \hat{W}_j\\
    &+\frac{1}{2}\sum_{j=1}^{N_V} L^0 V_j\ket{\psi_m}\delta \hat{W}_j\delta t+\frac{1}{2}\sum_{j_1,j_2=1}^{N_V} L^{j_1}V_{j_2}\ket{\psi_m}(\delta \hat{W}_{j_1}\delta\hat{W}_{j_2}+U_{j_1,j_2}).   
\end{aligned}
\end{equation}
In this expression,  $\delta\hat{W}$ is an approximation of an increment of the Brownian motion. It was suggested in \cite{kloeden2011numerical} to sample with 
the three-point distribution,
\begin{equation}
    P(\delta\hat{W} = \pm\sqrt{3\Delta}) = \frac{1}{6}, P(\delta\hat{W} = 0) = \frac{2}{3}
\end{equation}

Meanwhile, $U_{j_1,j_2}$'s are independent two-point distributed random variables with
\begin{equation}
\begin{aligned}
    &P(U_{j_1,j_2} = \pm\Delta) = \frac{1}{2},\forall j_2=1,\cdots,j_1-1\\
&U_{j_1,j_1} = -\Delta\\
&U_{j_1,j_2} = -U_{j_2,j_1}, \text{ for } j_2 = j_1+1,\cdots,N_V, j_1=1,\cdots,N_V.
\end{aligned}
\end{equation}

In \cref{sde2nd}, the diffusion operator $L^{j_1}$'s are defined as follows. 
\begin{equation}
    \begin{aligned}
        L^0V_j\ket{\psi} &= \sum_{k=1}^d (G\psi)^k\frac{\partial}{\partial\psi^k}V_j\ket{\psi}=\sum_{k=1}^d (G\psi)^kV_j(:,k)= V_jG\ket{\psi} \\
        L^{j_1}V_{j_2}\ket{\psi} &= \sum_{k=1}^d (V_{j_1}\psi)^k\frac{\partial}{\partial\psi^k}V_{j_2}\ket{\psi} =\sum_{k=1}^d (V_{j_1}\psi)^kV_{j_2}(:,k)
        =V_{j_2}V_{j_1}\ket{\psi}\\
    \end{aligned}
\end{equation}
Here the notation $(:,k)$ indicates the $k$th column of the matrix.  Thus, \cref{eq: 2nd-order} or \cref{2ndorder} can be rewritten as 
\begin{equation}
\begin{aligned}
    \ket{\psi_{m+1}} &= (I-\frac{1}{2}G\delta t)^{-1}(I+\frac{1}{2}G\delta t)\ket{\psi_m}+\sum_{j=1}^{N_V}(I-\frac{1}{2}G\delta t)^{-1}(V_j+\frac{1}{2}V_jG\delta t)\ket{\psi_m}\delta\hat{W}_j\\
    &+\frac{1}{2}(I-\frac{1}{2}G\delta t)^{-1}\sum_{j_1,j_2 = 1}^{N_V} V_{j_2}V_{j_1}\ket{\psi_m}(\delta \hat{W}_{j_1}\delta\hat{W}_{j_2}+U_{j_1,j_2}).
\end{aligned}
\end{equation}

By taking expectations, this time-marching scheme induces an  iteration formula for the density operator $\rho_m = \mathbb{E}[\ket{\psi_m}\bra{\psi_m}]$,  as follows.
\begin{equation}
\begin{aligned}
     \rho_{m+1} &= (I-\frac{1}{2}G\delta t)^{-1}(I+\frac{1}{2}G\delta t)\rho_m(I+\frac{1}{2}G\delta t)(I-\frac{1}{2}G\delta t)^{-1}   \\
     & +(I-\frac{1}{2}G\delta t)^{-1}\sum_{j=1}^{N_V}V_j(I+\frac{1}{2}G\delta t)\rho_m(I+\frac{1}{2}G\delta t)V_j^{\dag}(I-\frac{1}{2}G\delta t)^{-1}\delta t\\
     &+\frac{1}{2}(I-\frac{1}{2}G\delta t)^{-1}\sum_{j_1,j_2=1}^{N_V} V_{j_1}V_{j_2}\rho_m V_{j_2\dag}V_{j_1\dag}(I-\frac{1}{2}G\delta t)^{-1}(\delta t)^2\\
     &=\sum_{j=0}^{N_V}F_j\rho_mF_j^{\dag}+\sum_{j_1,j_2=1}^{N_V}A_{j_1,j_2}\rho_mA_{j_1,j_2}^{\dag}\\
     & = \sum_{j=0}^{N_V^2+N_V}F_j\rho_mF_j^{\dag}= \mathcal{K}[\rho_m]
\end{aligned}
\end{equation}
where the Kraus operator 
\begin{equation}
    \begin{aligned}
        F_0 &= (I-\frac{1}{2}G\delta t)^{-1}(I+\frac{1}{2}G\delta t)\\
        F_j & = (I-\frac{1}{2}G\delta t)^{-1}V_j(I+\frac{1}{2}G\delta t)\sqrt{\delta t},j = 1,\cdots,N_V\\
        A_{j_1,j_2} & = \frac{1}{\sqrt{2}}(I-\frac{1}{2}G\delta t)^{-1} V_{j_1}V_{j_2}\delta t, j_1,j_2 =1,\cdots, N_V\\
        F_{j_1+N_Vj_2}&=A_{j_1,j_2}, j_1,j_2 = 1,\cdots, N_V\\
    \end{aligned}
\end{equation}

To show that the method has second order accuracy, we first expand the
solution of the Lindblad equation \eqref{eq: Lindblad1},
\begin{align}\label{expd-lindblad}
    \rho(\tau_{m+1}) = \rho(\tau_m) + \delta t \mathcal{L} \rho(\tau_m) + \frac{1}{2}  \delta t^2 \mathcal{L}^2 \rho(\tau_m) +\mathcal{O}(\delta t^3).
\end{align}
In particular, we have,
\[\mathcal{L} \rho = -i[H, \rho] + \sum_{j=1}^{N_V} \left(V_j\rho V_j^{\dag}-\frac{1}{2} \{ V_j^{\dag}V_j,\rho\} \right).
\]

We now show that the method \eqref{eq: 2nd-order} has a consistent expansion at \eqref{expd-lindblad} up to $\mathcal{O}(\delta t^2).$ Toward this end, we expand the Kraus operators,
\[
\begin{aligned}
    &F_0=  I + G \delta t + \frac12 G^2 \delta t^2 + \mathcal{O}(\delta t^3), \\ 
    &F_j =  V_j \sqrt{\delta t}  + \frac12 \{G, V_j \} \delta t \sqrt{\delta t}
      + \frac12 \{G,G V_j\} \delta t^2 \sqrt{\delta t}
     + \mathcal{O}(\delta t^{7/2}),\qquad j = 1,\cdots,N_V\\
       &F_{j_1+N_Vj_2}= \frac{1}{\sqrt{2}} V_{j_1}V_{j_2}\delta t 
       + \mathcal{O}(\delta t^{3/2}),\qquad j_1,j_2 = 1,\cdots,N_V
\end{aligned}
\]

By substituting these expansions of the jump operators, and only keeping terms of order up to $\mathcal{O}(\delta t^{2})$, we find that,
\begin{align*}
\rho_{m+1}= &\rho_m + \delta t (G \rho_m + \rho_m G^\dag) + \frac{1}{2} \delta t^2 (G^2 \rho_m + 2 G\rho_m G^\dag + \rho_m {G^\dag}^2 ) \\     
&  + \delta t \sum_{j=1}^{N_V} V_j \rho_m V_j^\dag 
  + \frac{1}{2} \delta t^2  \sum_{j=1}^{N_V} \left( V_j \rho_m \{G, V_j \}^\dag   + \{G, V_j \}  \rho_m V_j^\dag \right)\\
 &  + \frac{1}{2} \delta t^2   \sum_{j_1, j_2=1}^{N_V} V_{j_1}V_{j_2} 
  \rho_m V_{j_2}^\dag V_{j_1}^\dagger  + \mathcal{O}(\delta t^3). 
\end{align*}

In light of \cref{matG}, the $\mathcal{O}(\delta t)$ terms are consistent with that in \eqref{expd-lindblad}. With lengthy calculations, we can also see that the $\mathcal{O}(\delta t^2)$ terms are also consistent. From the calculation of the consistency of $\mathcal{O}(\delta t)$, we know that $G\rho_m+\rho_m G^{\dag} = \mathcal{L}\rho_m-\sum_jV_j\rho_m V_j^{\dag}$.
The detailed calculation can be summarized as follows:
\begin{equation*}
    \begin{aligned}
        &G^2 \rho_m + 2 G\rho_m G^\dag + \rho_m {G^\dag}^2  = G(G\rho_m+\rho_mG^{\dag})+(G\rho_m+\rho_m G^{\dag})G^{\dag} \\
        =& \mathcal{L}^2\rho_m-\sum_{j=1}^{N_V}\mathcal{L}(V_j\rho_m V_j^{\dag})-\sum_{j=1}^{N_V}V_j\mathcal{L}\rho_m V_j^{\dag}+\sum_{j_1,j_2 = 1}^{N_V} V_{j_2} V_{j_1}\rho_m V_{j_1}V_{j_2}\\
        & \sum_{j=1}^{N_V} \left( V_j \rho_m \{G, V_j \}^\dag   + \{G, V_j \}  \rho_m V_j^\dag \right) = \sum_j (V_j\rho_m V_j^{\dag}G^{\dag}+GV_j\rho_mV_j^{\dag})+V_j(\rho_m G^{\dag}+G\rho_m)V_j^{\dag}  \\
        =& \sum_j\mathcal{L}(V_j\rho_mV_j^{\dag})-2\sum_{j_1,j_2} V_{j_2}V_{j_1}\rho_m V_{j_1}^{\dag}V_{j_2}^{\dag}+\sum_j V_j\mathcal{L}\rho_m V_j^{\dag}\\
        \Rightarrow & G^2 \rho_m + 2 G\rho_m G^\dag + \rho_m {G^\dag}^2+\sum_{j=1}^{N_V} \left( V_j \rho_m \{G, V_j \}^\dag   + \{G, V_j \}  \rho_m V_j^\dag \right)+\sum_{j_1,j_2} V_{j_2}V_{j_1}\rho_m V_{j_1}^{\dag}V_{j_2}^{\dag} = \mathcal{L}^2\rho_m.
    \end{aligned}
\end{equation*}
This completes the proof.

\end{document}